\documentclass[prx,nofootinbib,amsmath,amssymb, notitlepage, twocolumn,superscriptaddress,longbibliography]{revtex4-1}
\usepackage[pdftex,colorlinks=true,linkcolor=darkblue,citecolor=blue,urlcolor=darkred]{hyperref}
\usepackage{braket}
\usepackage{amsmath,amsfonts, amssymb, amsthm, dsfont}
\usepackage{yfonts,bbm}
\usepackage{bm,breqn}
\usepackage{mathrsfs}
\usepackage{graphicx}
\usepackage[caption=false]{subfig}
\usepackage{verbatim}
\usepackage{xcolor}

\definecolor{darkblue}{rgb}{0.,0.,0.4}
\definecolor{darkred}{rgb}{0.5,0.,0.}

\usepackage{verbatim}
\usepackage{stackengine}

\usepackage{tikz}
\usetikzlibrary{calc}
\usepackage{multirow}
\usepackage[capitalise]{cleveref}

\newcommand{\refeq}[1]{Eq.~(\ref{#1})}
\newcommand{\reffig}[1]{Fig.~\ref{#1}}

\newcommand{\refsec}[1]{Sec.~\ref{#1}}

\renewcommand{\d}{\mathrm{d}}
\newcommand{\E}{\mathbb{E}}
\newcommand{\tr}{\text{tr}}
\newcommand{\stabgen}{\mathcal{S}}

\newcommand{\identity}{\mathbb{I}}
\newcommand{\pauli}{\mathcal{P}}
\newcommand{\hilbert}{\mathcal{H}}
\newcommand{\channel}{\mathcal{E}}
\newcommand{\opspace}{L}
\newcommand{\supp}{\mathsf{supp}}

\newcommand{\eqfig}[2]{\vcenter{\hbox{\includegraphics[height=#1]{#2}}}}

\newcommand{\add}[1]{#1}
\newcommand{\change}[1]{}

\newtheorem{theorem}{Theorem}
\newtheorem{lemma}{Lemma}

\begin{document}
\title{Ultrafast Entanglement Dynamics in Monitored Quantum Circuits}
\author{Shengqi Sang}\affiliation{\PI}\affiliation{\UW}
\author{Zhi Li}\affiliation{\PI}
\author{Timothy H. Hsieh}\affiliation{\PI}
\author{Beni Yoshida}\affiliation{\PI}
\newcommand*{\PI}{Perimeter Institute for Theoretical Physics, Waterloo, Ontario N2L 2Y5, Canada}
\newcommand*{\UW}{Department of Physics and Astronomy, University of Waterloo, Waterloo, Ontario N2L 3G1, Canada}

\begin{abstract}
% \add{Revise}
    Projective measurement, a basic operation in quantum mechanics, can induce seemingly nonlocal effects. In this work, we analyze such effects in many-body systems by studying the non-equilibrium dynamics of weakly monitored quantum circuits, focusing on entanglement generation and information spreading. We find that, due to measurements, the entanglement dynamics in monitored circuits is indeed ``faster" than that of unitary ones in several ways. 
    Specifically, we find that a pair of well-separated regions can become entangled in a time scale $\ell^{2/3}$, sub-linear in their distance $\ell$. For the case of Clifford monitored circuits, this originates from super-ballistically growing stabilizer generators of the evolving state.
    In addition, we find initially local information can spread super-ballistically as $t^{3/2}$. 
    % We show that the super-ballistic entanglement dynamics is closely tied with the dynamical encoding process
    Furthermore, by viewing the dynamics as a dynamical encoding process, we show that the super-ballistic growing length scale relates to an encoding time that is sublinear in system size. 
    % \add{As a theoretical contribution, we propose a way to generalize operator dynamics in unitary settings to non-unitary ones}
    \add{To quantify the information dynamics, we develop a formalism generalizing operator spreading to non-unitary dynamics, which is of independent interest.}
\end{abstract}
\maketitle

% \tableofcontents
\section{Introduction}
Causality, which places an ultimate speed limit on the propagation of information, is a cornerstone of physics. 
Even for non-relativistic many-body quantum systems, there is an emergent speed limit, as long as the microscopic interactions are short-ranged and bounded~\cite{lieb1972finite}.
Such emergent causality is manifest in various different aspects of many-body quantum dynamics, including correlation functions, entanglement growth, and operator spreading 
\cite{shenker2014black, roberts2015localized, hosur2016chaos, calabrese2005evolution, nahum2017entanglement, von2018operator, nahum2018operator}. 

In the non-relativistic limit, however, projective measurements can have seemingly causality-violating effects on quantum correlations. 
This property underpins many fundamental quantum information protocols, such as quantum teleportation and entanglement swapping, which transmit quantum information among well-separated parties ``instantly" once the measurement outcomes are transmitted. 
Similar phenomena also appear in many-body physics \cite{friedman2022locality, bao2021finite} and can be exploited for practical advantages such as preparing long-range entangled states efficiently \cite{PhysRevLett.127.220503,https://doi.org/10.48550/arxiv.2209.06202,https://doi.org/10.48550/arxiv.2205.01933,lu2022measurement, lee2022decoding, tantivasadakarn2022hierarchy, zhu2022nishimori}, which would be prohibited by the emergent causality if only unitary operations are allowed.

Given these facts, we are naturally led to ask how the causal dynamics in unitary many-body systems are modified by projective measurements. An ideal setting to address such questions is monitored quantum circuits, a subject of many recent studies. By arranging local unitaries and measurements in a random fashion, random monitored circuits provide a useful minimal tractable model for generic many-body systems under hybrid (unitaries and measurements) dynamics. This model has been studied extensively and found to display a rich set of phenomena including a dynamical phase transition driven by the frequency of measurements \cite{aharonov2000quantum,li2018zeno,chan2019unitary,skinner2019measurement, gullans2020dynamical, li2019measurement, li2021conformal, zabalo2020critical,zabalo2022operator,jian2020measurement,bao2020theory,gullans2020scalable}. 

In this work, we revisit the dynamics of monitored circuits with low measurement frequency, namely weakly monitored quantum circuits, from the perspective of emergent nonlocality due to projective measurements. We find that the existence of measurements drastically modifies the dynamics of the entanglement. Specifically, we consider the following two complementary aspects of entanglement dynamics.
\begin{enumerate}
    \item[a)] How does the entanglement structure evolve with time starting from an unentangled (maximally mixed or pure product) initial state?
    \item[b)] How does the monitored dynamics spread and forget the input information?
\end{enumerate}
In both settings, we find that the entanglement and information dynamics exhibit ultrafast behavior due to projective measurements. Specifically, in the first setting, the time for two well-separated regions to build non-trivial multipartite entanglement has a sublinear scaling $t\sim(\text{distance})^{2/3}$. 
In the second setting, initially localized information spreads out to the super-linear length scale as $ \ell \sim(\text{time})^{3/2}$. In contrast, in unitary dynamics, any length scale can grow at most linearly in time due to emergent causality. 

We also discuss our results from the perspective of viewing monitored circuits as dynamical encoding processes \cite{gullans2020dynamical, li2021statistical, yoshida2021decoding, fidkowski2021dynamical}.
We find that the encoding is completed in a time scale $t\sim L^{2/3}$ with $L$ being the total system size. Furthermore, we find that the code distance grows as $d_{\text{code}}\sim t^{1/2}$ until saturation; and relatedly, the input state's local information contained in size $\ell$ is lost at a time scale $t\sim \ell^2$. 

Finally, we discuss the relation between our work and directed polymers in a random environment (DRPE)~\cite{huse1985pinning, PhysRevLett.58.2087}, which was recently proposed to be an effective description of entanglement dynamics in weakly-monitored circuits \cite{li2021dpre}. 

\subsection{Setup}

We consider one-dimensional (1D) random monitored quantum dynamics generated by two-qubit random unitary Clifford gates and randomly applied single-qubit measurements in the Pauli $Z$ basis, as depicted in \reffig{fig:circuit}. The restrictions to Clifford gates and Pauli measurements, as well as suitable initial states, ensure efficient classical simulation using the stabilizer formalism \cite{gottesman1997stabilizer}.

When a state $\rho$ is measured in the $Z$ basis, the (unnormalized) resulting state will be $\Pi_{\pm}\rho\Pi_{\pm}$ with a probability $p_{\pm}=\tr(\Pi_{\pm}\rho)$, where $\Pi_{\pm}$ is the projector to the $|0\rangle/|1\rangle$ subspace. 
For a fixed circuit structure and input state $\rho_0$, the evolving trajectory state $\rho_{\textbf{m}}(t)$ depends on all the earlier measurement outcomes $\mathbf{m}$:
\begin{equation}\label{eq:evolution}
    \rho_{\mathbf{m}}(t)\propto C_{\mathbf{m}}(t) \rho_0 C^\dagger_{\mathbf{m}}(t),
\end{equation}
where $C_{\textbf{m}}(t)$ is a product of projectors and unitary operators. 
The probability of getting this output is $p_{\textbf{m}}=\tr(C_{\mathbf{m}} \rho_0 C_{\mathbf{m}}^\dagger)$. We omit the subscript $\textbf{m}$ whenever the meaning is clear from the context.
Randomness in this monitored circuit is three-fold: choices of two-qubit unitary Clifford gates, locations of measurements, and measurement outcomes. 
We will be interested in averaged behavior over all three sources of randomness.

\begin{figure}
\centering
\includegraphics[width=0.25\textwidth]{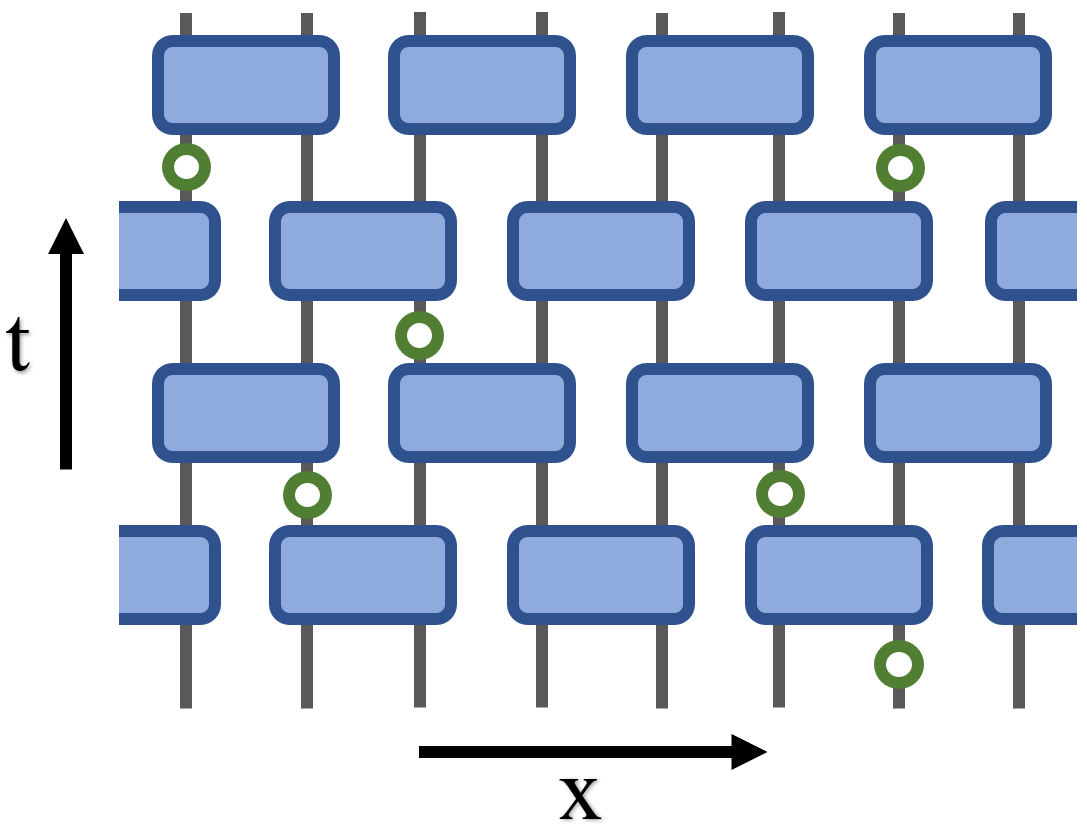}

\caption{An illustration of monitored quantum circuits. Each blue block represents a random unitary gate drawn from the 2-qubit Clifford group. On each vertical bond, a projective measurement (represented as a green dot) in the Pauli-$Z$ basis is applied with a probability $p$. 
}
\label{fig:circuit}
\end{figure}

In this work, we focus on the properties of the weakly-monitored phase. In all the simulations presented later, we take $p=0.08\approx 0.5 p_c$ where $p_c\approx 0.16$ is the critical point for the measurement-induced phase transition~\cite{li2019measurement}. Nevertheless, we expect our results concerning the entanglement dynamics are universal for monitored circuits in the weakly monitored phase with disorder, see \refsec{sec:discussion} for more discussions.

\subsection{Summary of Results}

% \add{Revisit later}

\textbf{Entanglement dynamics:}
In \refsec{sec:entgen}, we study how the entanglement structure of a trajectory state $\rho_{\textbf{m}}(t)$ evolves with time until its saturation.  

\emph{Stabilizer growth:} 
We start by considering the dynamics with a maximally mixed initial state. 
The length distribution of the stabilizer generators provides an informative proxy for the underlying entanglement structure.
In \refsec{sec:clip}, we find that the evolution of stabilizer length distribution takes the following form:
\begin{equation}\label{eq:distr-sum}
    h(\ell,t) \simeq \frac{e^{-\ell/\ell^*(t)}}{\ell^{5/3}},
\end{equation}
with a time-dependent soft cutoff $\ell^*(t)$ which exhibits a super-linear growth: 
\begin{equation}\label{eq:lstar1}
    \ell^*(t)\propto t^{\frac{3}{2}},
\end{equation} 
indicating that there is a small fraction of stabilizers that grows super-linearly with time, as conjectured in~\cite{yoshida2021decoding}. 

\emph{Entanglement growth:} 
The fast-growing stabilizers leave important imprints on $\rho(t)$'s entanglement evolution. In \refsec{sec:CMI}, we show that they can be detected by a tripartite entanglement measure, namely the conditional mutual information:
\begin{equation}
    I_{A:B|C}\equiv S_{AC}+S_{BC}-S_{ABC}-S_{C},
\end{equation}
where $A$ and $B$ are two regions separated by another large region $C$. We show that $I_{A:B|C}$ exactly measures the number of nonlocal stabilizer generators connecting $A$ and $B$. \refeq{eq:distr-sum} then implies that $I_{A:B|C}$ starts to take non-zero values at a time scale
\begin{equation}\label{eq:intro32}
    t^* \simeq \mathrm{dist}(A,B)^{\frac{2}{3}};
\end{equation}
which suggests that, at time $t$, non-trivial entanglement at the distance scale of $O(t^{3/2})$ will be generated.

\emph{Code growth:} 
The super-linear growth of entanglement naturally gives rise to a new sublinear time scale of $O(L^{2/3})$ when the cutoff $\ell^*(t)$ reaches the system size $L$. In \refsec{sec:coded}, we argue that the $O(L^{2/3})$ time scale coincides with the encoding time of a monitored circuit when it is regarded as a dynamically generated quantum error-correcting code whose code subspace is supported on $\rho(t)$. 
Furthermore, we show that the $\rho(t)$'s (contiguous) code distance grows as
\begin{equation}\label{eq:dcode1}
    d_{\text{code}}\propto t^{\frac{1}{2}},
\end{equation}
until it reaches its steady value $\simeq O(L^{1/3})$~\cite{li2021statistical}
% \tim{do we need to cite Yaodong and Matthew's paper here?}
, which occurs at $t \simeq O(L^{2/3})$.

\emph{Pure initial state:} 
In \refsec{sec:pure_vs_mix}, we extend the discussion to the dynamics starting from a pure initial state. We find that the entanglement dynamics is qualitatively similar to that of unitary-only dynamics and does not display any super-linearly growing length scale. \\

\begin{figure}
% \captionsetup[subfigure]{aboveskip=-2pt,belowskip=-2pt}
\centering
\includegraphics[width=0.35\textwidth]{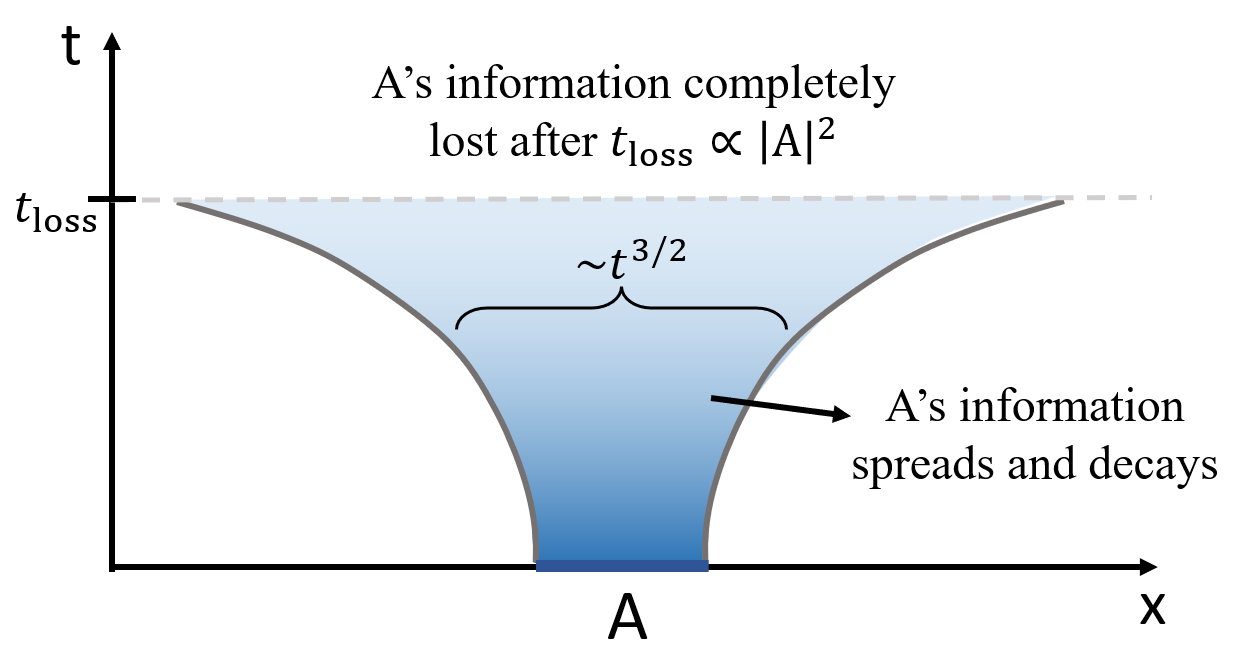}

\caption{An illustration of the spreading and decay of local information.
}
\label{fig:illu_infospread}
\end{figure}

\textbf{Information dynamics:}
In \refsec{sec:info-spread}, we switch our attention from an evolving trajectory state $\rho(t)\propto C(t)\rho_0 C^\dagger(t)$ to the circuit itself and study the fate of initially local information. To quantify information spreading, we introduce a reference system $R$ which is initially maximally entangled with the physical qubits in the circuit, referred to as $P$ (\reffig{fig:two_copy}). After acting the monitored dynamics on $P$ for time $t$, we use the following quantity to measure how much information within region $A_P$ at input is detectable within region $B_P$ at output:
\begin{equation}
I_{A_R: B_P}(t) \equiv S_{A_R}(t) + S_{B_P}(t)- S_{A_R\cup B_P}(t),
\end{equation}
where $A_R$ is the part of the reference system corresponding to $A_P$. 
%In \refsec{sec:meaning}, we provide two physical meanings [\refeq{eq:recoverlogical} and \refeq{eq:monitored_relation}] of this quantity. 

\add{In \refsec{sec:Troperators}, we show that any Clifford monitored dynamics can be treated as a map between two quantum code spaces. We introduce the notion of \textit{logical operator transferring}, which is a natural generalization of operator evolution in unitary dynamics. Furthermore, we show that $I_{A_R: B_P}(t)$ exactly counts how many logical operators are transferred from region $A$ to the region $B$ by the dynamics $C(t)$ (Thm.\ref{thm:transferable}).}

\add{Then we focus on the behavior of information spreading in a monitored circuit.} Our findings are graphically illustrated in \reffig{fig:illu_infospread}. Roughly, local information undergoes a combination of two types of dynamics, decay and spreading, as summarized in the heat map.

\emph{Decay of local information:}
Measurements can destroy the input information. In \refsec{sec:forget}, we show that the information within an interval $A$ will be completely destroyed at a time scale:
\begin{equation}
    t_{\text{loss}}\propto |A|^2,
\end{equation}
with $I_{A_R: P}(t)=0$ for $t>t_{\text{loss}}$. In fact, this relation is fundamentally equivalent to \refeq{eq:dcode1} due to a statistical time-reversal symmetry of the dynamics, as further discussed in \refsec{sec:forget}.

\emph{Spreading of local information:}
Before the time $t_{\text{loss}}$, local information gradually delocalizes. 
In \refsec{sec:wavefront}, we numerically find that the remaining part of initially local information is contained within a region of the size
\begin{equation}
    D(t)\propto t^{\frac{3}{2}}.
\end{equation}
This super-ballistic information spreading again originates from the nonlocal effects of measurements. We also confirm that it is ballistic in unitary circuits. 

% \add{Add 3/2 for time-reversal here, or somewhere before.}

\emph{Initial state dependence:}
In \refsec{sec:alt-way}, we discuss how the behavior of information spreading changes if we fix a part of the input state. Specifically, we consider the setting where only a subsystem $A\subseteq P$ is considered as the \textit{variable} input, while the rest of the system $A^c=P-A$ has a \textit{fixed} initial state.
We find that the spreading of $A$'s information depends on the choice of the initial state for $A^c$. Namely, if $A^c$ is initially in a pure product state, then the spreading becomes linear instead of super-linear.\\

\textbf{Comparison with DRPE:}
In \refsec{sec:domain-wall}, we discuss the relation between our results and the statistical mechanics of directed polymers in random environment (DRPE), an effective theory recently proposed for the entanglement dynamics in weakly monitored circuits \cite{li2021dpre}. We demonstrate that some of our results can be derived using DPRE, while the rest are numerically consistent with DPRE.

\section{Entanglement Dynamics}\label{sec:entgen}

In this section, we study the dynamics of the stabilizer length distribution in weakly monitored circuits and then discuss its implications for the evolving state's entanglement structure and the contiguous code distance growth.
We will also discuss the initial state dependence of these properties.

\subsection{Stabilizer length distribution in clipped gauge}\label{sec:clip}

An $L$-qubit stabilizer state $\rho$, associated with a set of mutually commuting Pauli operators (dubbed stabilizers) $\stabgen=\{g_1,g_2,...g_m\}$, is defined as:
\begin{equation}
    \rho = \frac{1}{2^L}\sum_{g\in G}g,
\end{equation}
where $G=\braket{\stabgen}$ is the stabilizer group: the Abelian group generated by all possible products of the operators in $\stabgen$. If $m=L$, then the state $\rho$ will be a pure state.

Since two different sets of stabilizer generators could generate the same stabilizer group,  there is a gauge degree of freedom in choosing $\stabgen$ for a given state $\rho$. When the qubits are arranged in a 1D chain, there is a gauge choice called \textit{clipped gauge} \cite{nahum2017entanglement, li2019measurement}, that is particularly suitable for studying the entanglement structure. The clipped gauge is defined as follows. We denote the left endpoint and right endpoint of each generator $g\in \stabgen$, a Pauli string operator, as $l(g)$ and $r(g)$. We say $\stabgen$ is in the clipped gauge if the following two conditions are satisfied:
\begin{itemize}
    \item $|\{g\in\stabgen:\ l(g)=i\}| + |\{g\in\stabgen:\ r(g)=i\}| \leq 2\quad \forall i$;
\item if one of the two terms above equals 2, then two left (or right) endpoints at position $i$ must be different Pauli operators.
\end{itemize}
For any stabilizer state $\rho$, a clipped-gauged $\stabgen$ always exists, but may not be unique. However, the set of end-points pairs defined below is unique
\begin{equation}\label{eq:paris}
    \mathcal{B}(\rho)\equiv \{(l(g), r(g)):\ g\in \stabgen\}.
\end{equation}
See the appendix of Ref.~\cite{li2019measurement} for details.

The clipped gauge and the set $\mathcal{B}$ are particularly useful for studying the entanglement structure. Specifically, for any contiguous region $A$, we have \cite{li2019measurement,nahum2017entanglement}: 
\begin{equation}\label{eq:clipS}
    S_A=|A|-|\{g\in\stabgen: \supp(g)\subseteq A\}|;
\end{equation}
where $\supp(g)$ is the interval $[l(g), r(g)]$. Therefore for two neighboring contiguous regions $A$ and $B$, 
\begin{equation}\label{eq:clipMI}
\begin{aligned}
    I_{A:B} 
    &\equiv S_A + S_B - S_{AB}\\
    &=|\{g\in\stabgen: l(g)\in A, r(g)\in B\}|.
\end{aligned}
\end{equation}
Namely, for questions involving contiguous regions only, clipped-gauged stabilizer generators can be thought of as ``generators of entanglement" where entanglement across a cut is proportional to the number of bridging generators. 

Even when $A$ and $B$ are not neighboring, a certain useful entanglement measure can be explicitly computed. Suppose that $A$ and $B$ are separated by some interval $C$. Using \refeq{eq:clipS} four times, we find that the conditional mutual information is given by
\begin{equation}\label{eq:clipCMI}
\begin{aligned}
    I_{A:B|C}
    &\equiv I_{A:BC}-I_{A:C}\\
    &=
    S_{AC} + S_{BC} - S_{ABC} - S_{C}\\
    &=|\{g\in\stabgen: l(g)\in A, r(g)\in B\}|.\\
\end{aligned}
\end{equation}
This relation tells us $I_{A:B|C}$ directly measures the sizes of the stabilizers. Namely, the number of long stabilizers starting from $A$ and ending in $B$ exactly equals the conditional mutual information. 

Returning to the setting of (1+1)D monitored quantum circuits, for the evolving stabilizer state $\rho(t)$, we define its stabilizer length distribution as:
\begin{equation}\label{eq:len_dist}
    h(\ell,t, x) = |\{g\in\stabgen(\rho(t)): \mathsf{len}(g)=\ell,\ \mathsf{mid}(g)=x\}|,
\end{equation}
where $\stabgen(\rho(t))$ is a clipped-gauged set of generators for $\rho(t)$, while $\mathsf{len}(\cdot)$ and $\mathsf{mid}(\cdot)$ are defined as:
\begin{equation}
\begin{aligned}
    \mathsf{len}(g) &= r(g)-l(g)+1\\
    \mathsf{mid}(g) &= \left\lfloor \tfrac{r(g)+l(g)}{2} \right\rfloor.
\end{aligned}
\end{equation}
Evidently, $h(\ell,t, x)$ only depends on $\mathcal{B}(\rho(t))$, thus is well-defined for a given $\rho(t)$.
In the thermodynamic limit $L\rightarrow \infty$, the distribution should be independent of $x$ as the dynamics is statistically translational invariant. In numerical simulations where system size is always finite, we expect this to still be the case, especially when the typical stabilizers' endpoints $x\pm \ell/2$ are far from the system's boundary. We omit the $x$ dependence of $h$ from now on.

\begin{figure}[h]
% \captionsetup[subfigure]{aboveskip=-2pt,belowskip=-2pt}
\centering{
\subfloat[]{\includegraphics[width=0.40\textwidth]{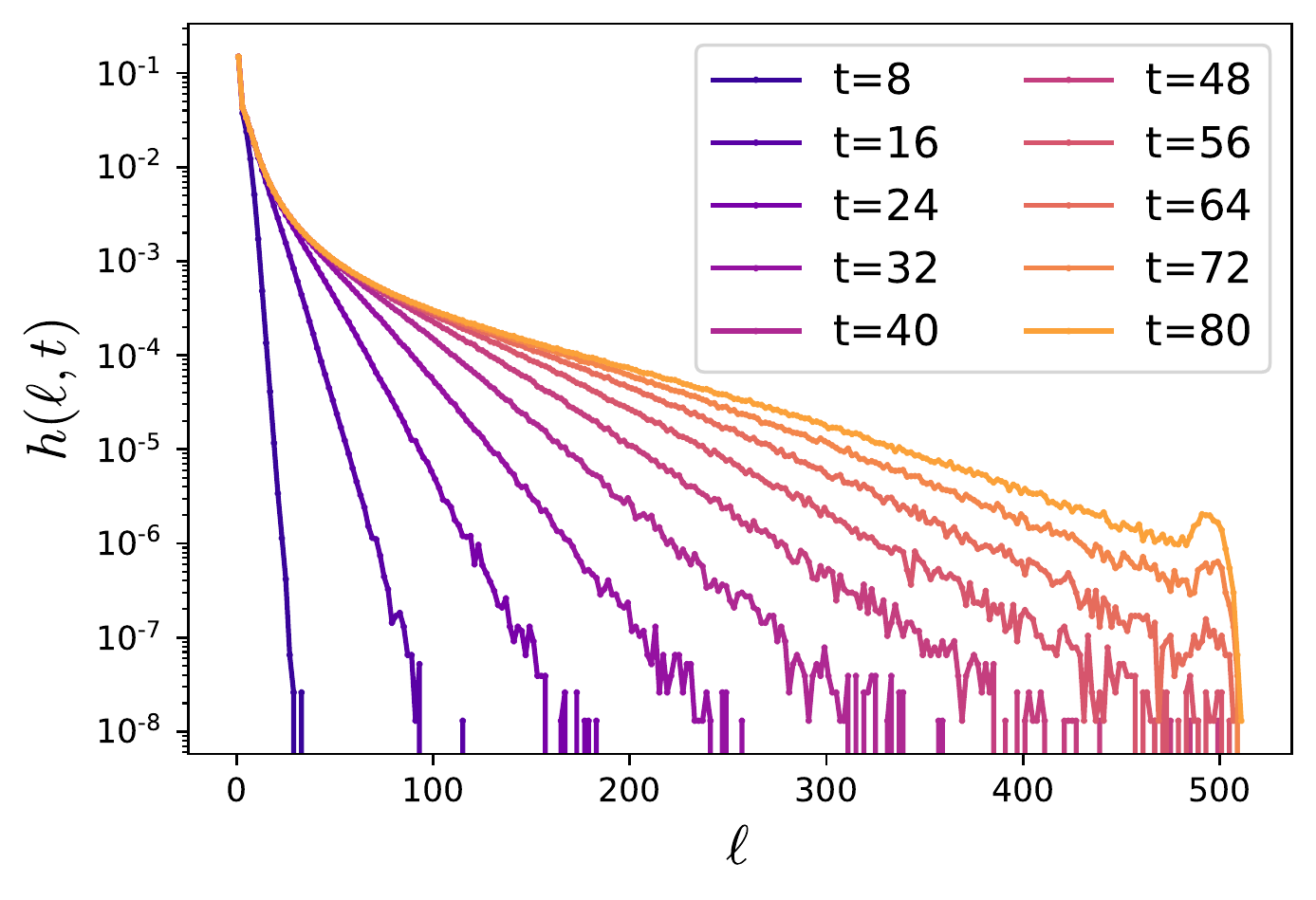}}\\
\subfloat[]{
\stackinset{r}{0.36\linewidth}{b}{0.13\linewidth}
{\includegraphics[width=0.14\textwidth]{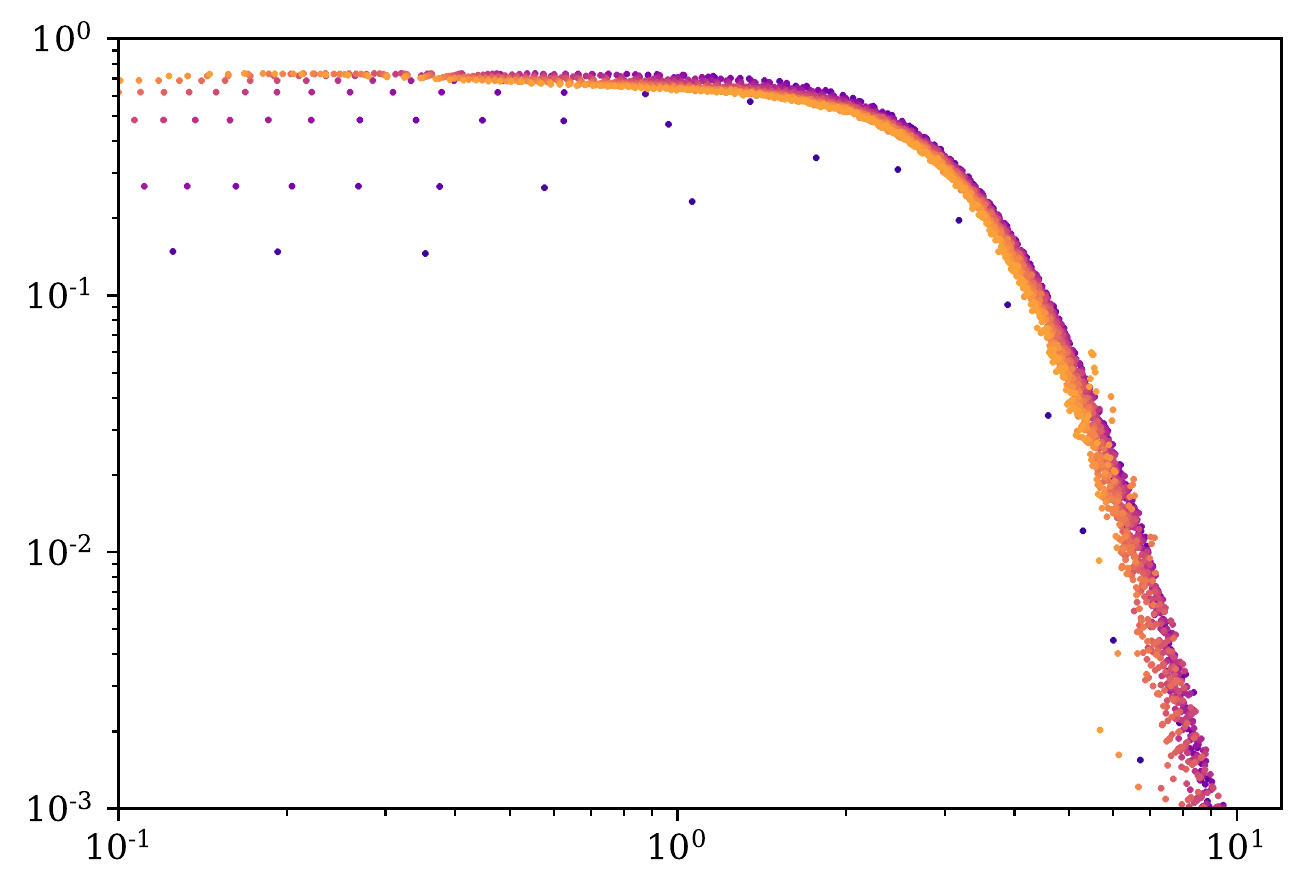}}
{\includegraphics[width=0.37\textwidth]{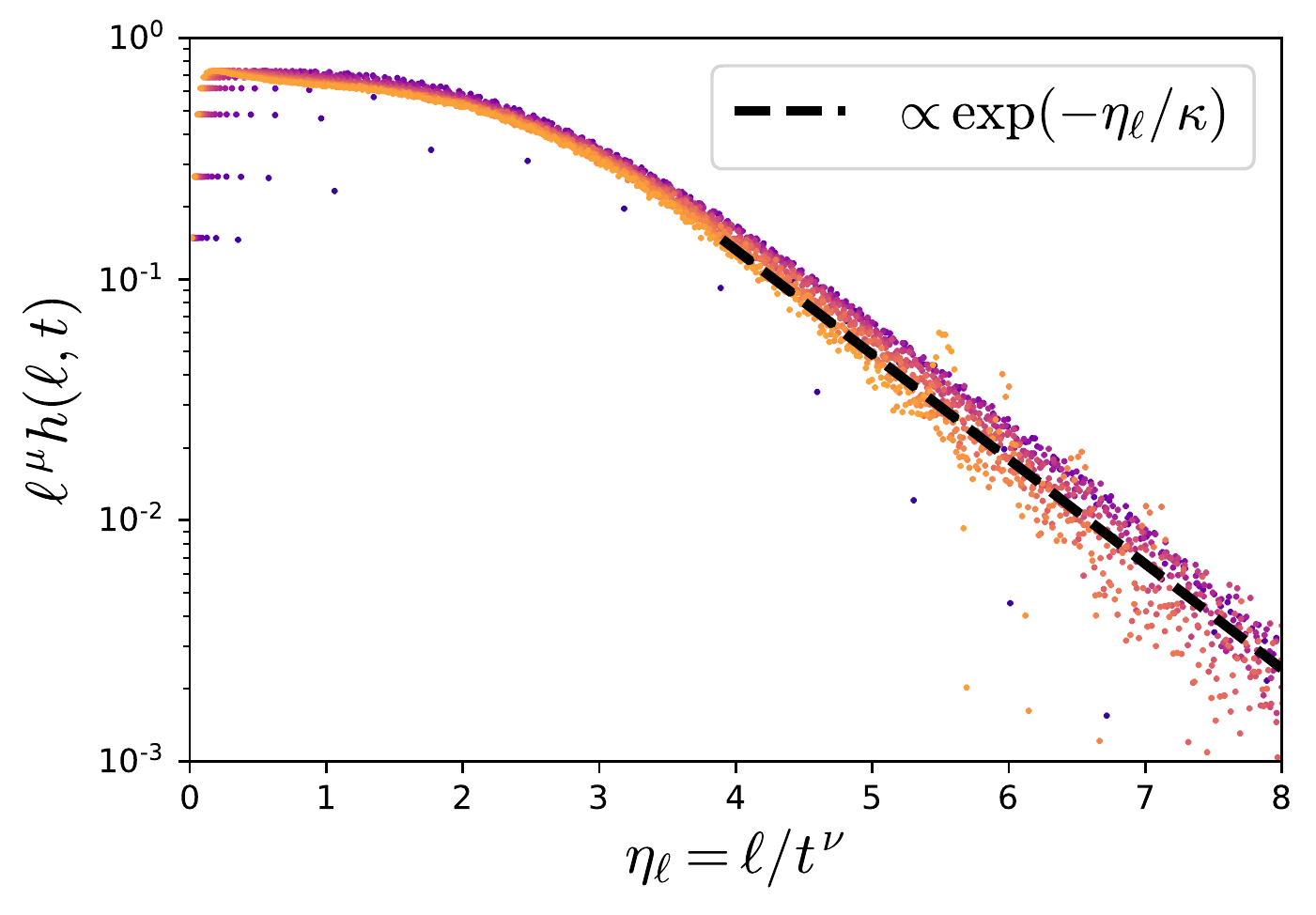}}
}
}
\caption{
(a) Stabilizer length distribution $h(\ell,t)$ in an $L=512$ circuit with open boundary condition. (b) The same data after rescaling the x and y axes by $t^{-\nu}$ and $t^\mu$, respectively. The inset shows the same plot on a log-log scale.}
\label{fig:stab_dist}
\end{figure}

We now study how the distribution evolves with time before the system reaches steady states with a maximally mixed initial state $\rho(0)=\frac{1}{2^L}\mathbb{I}$.
The simulation result, shown in \reffig{fig:stab_dist}(a), leads us to the following scaling form for $h(\ell, t)$:
\begin{equation}\label{eq:length_dist_scale}
    h(\ell,t) = c(\ell/\ell^*(t)) \frac{e^{-\ell/\ell^*(t)}}{\ell^{\mu}},
\end{equation}
with $c(x)$ is a slow varying function that takes non-vanishing $O(1)$ value for any $x\in(0,\infty)$. 
Namely, below the length scale $\ell^*(t)$, the distribution has an equilibrium form that possesses a time-independent power-law tail. 
On the other hand, above the length scale $\ell^*(t)$, the distribution exhibits an exponential tail, indicating that stabilizer generators of size larger than $\ell^*(t)$ are extremely rare and are essentially negligible.

Furthermore, we find that the cutoff $\ell^*(t)$ grows according to a power law in time: 
\begin{equation}\label{eq:lengthscale}
\ell^*(t) \propto t^\nu. 
\end{equation}
This observation, together with the validity of the scaling form in Eq.~\eqref{eq:length_dist_scale}, is reflected by the data collapse of $h(\ell,t)$ at different time $t$,  shown in \reffig{fig:stab_dist}(b). The best collapse is given by $\mu=1.66$ and $\nu=1.5$, telling us that $\ell^*(t)\propto t^\nu$ grows super-linearly with time.
As we will discuss later in \refsec{sec:domain-wall}, these results are consistent with the domain-wall picture which predicts $\mu = \frac{5}{3}$ and $\nu = \frac{3}{2}$. 

These results suggest that a small but non-negligible fraction of stabilizer generators will grow super-linearly, namely $\sim t^{3/2}$. We will later demonstrate that these fast-growing stabilizer generators can be indeed detected by a multipartite entanglement measure, namely the conditional mutual information $I_{A:B|C}$.

\subsection{Entropy of a contiguous region}\label{sec:S_cont}

We have seen that the distribution $h(\ell,t)$ possesses a super-linear length scale $\ell^*(t)\propto t^{\nu}$. Given the close relation between $h(\ell,t)$ and the $\rho(t)$'s entanglement structure, we expect the length scale to leave some imprint on the $\rho(t)$'s entanglement structure. In this subsection we analyze the simplest case: the entanglement entropy $S_{A}$ of a contiguous subsystem $A$.

Recalling that the stabilizer distribution $h(\ell,t)$ contains information about $\rho(t)$'s entanglement structure \refeq{eq:clipS}, we have: 
\begin{equation}
\begin{aligned}\label{eq:vN}
    S_A(t) 
    &\equiv |A|-|\{g\in\stabgen: \supp(g)\subseteq A\}|\\
    &= |A| - \sum_{l=1}^{|A|}(|A|-\ell)h(\ell,t)\\
    &\approx |A| - \int_{a}^{|A|}(|A|-\ell)h(\ell,t) \d \ell.
\end{aligned}
\end{equation}
In the last (approximate) equality, we introduced a UV cutoff $a=O(1)$ in order to regularize the integral.

To build some intuition, let us begin by studying late-time cases with $\ell^*(t) \gg |A|$. Then, the distribution $h(\ell,t)$ in the above integral is proportional to $\ell^{-\mu}$. Performing the integral, we find that the entanglement entropy $S_{A}$ approaches to the following steady form for $t \gg |A|^{1/\nu}$:
\begin{align}\label{eq:SsmallA}
S_{A} \approx b_{1}|A| + b_{2} |A|^{2-\mu},
\end{align}
where $b_{1}, b_{2}\geq 0$ are some constants. 
Hence, we find that the entanglement entropy consists of the leading volume-law term as well as a sub-leading correction proportional to $|A|^{2-\mu}$ with $2-\mu \approx \frac{1}{3}$. Furthremore, we find that the saturation time $t\sim |A|^{1/\nu} \approx |A|^{\frac{2}{3}}$ is sublinear.  

Next, let us turn our attention to the early time cases with $\ell^*(t) \ll |A|$. The scaling behavior of the integral can be analyzed by approximating the distribution $h(\ell,t)$ as $\ell^{-\nu}$ with a cutoff at $\ell^*(t)$, which leads to the following result for $t \ll |A|^{1/\nu}$:
\begin{align}\label{eq:SlargeA}
S_{A} \approx b_{1}|A| + b_{2}^{(1)} |A| \ell^*(t)^{-(\mu-1)}   + b_{2}^{(2)} \ell^*(t)^{2-\mu},
\end{align}
with some constants $b_{1}, b_{2}^{(1)}, b_{2}^{(2)}> 0$. 

Here it is worth looking at the dynamical aspect of each term. The first term obviously represents the time-independent volume-law contribution which survives at late times as well. As for the second term, we notice that it is proportional to $|A|$, and hence contributes to the volume-law entanglement. As $t$ increases, however, the second term becomes smaller, reducing the volume-law contribution. Explicitly, the second term decays as 
\begin{align}\label{eq:b21}
 b_{2}^{(1)} |A| \ell^*(t)^{-(\mu-1)}  \sim c_1|A| t^{-\nu(\mu-1)}
\end{align}
until it becomes order of $|A|^{2-\mu}$ at $t \sim |A|^{1/\nu}$. Here it is natural to interpret this decay as the loss of the volume-law entanglement due to projective measurements in the bulk of the interval. As for the third term, it is initially zero, and then grows as 
\begin{align}\label{eq:b22}
b_{2}^{(2)} \ell^*(t)^{2-\mu} \sim c_2 t^{\nu(2-\mu)}
\end{align}
until it becomes order of $|A|^{2-\mu}$ at $t \sim |A|^{1/\nu}$. Since it does not depend on $|A|$, it is natural to interpret this growth as the creation of entanglement due to projective measurements followed by unitary evolution. 

Plugging $\mu=\frac{5}{3}$ and $\nu=\frac{3}{2}$ in, we obtain the following behavior for $S_A$ when $\ell^*(t)\ll |A|$:
\begin{equation}\label{eq:SlargeA2}
    S_A \approx (b_1+c_1 t^{-1}) |A| + c_2 t^{\frac{1}{2}}.
\end{equation}
It is worth emphasizing that, while both the $c_1$ and the $c_2$ terms eventually contribute the sub-leading terms $|A|^{1/3}$ at saturation when $t \sim |A|^{2/3}$, their origins are of different nature. It is also useful to note that the subleading term in the above equation implies that the mutual information $I_{A:B}$ between small neighboring regions $A,B$ grows sublinearly as $t^{1/2}$ at early time, which is consistent with an observation from \cite{li2021dpre}.

The time dependence of $S_A(t)$ at any time scale can also be studied by explicitly performing the integral \refeq{eq:length_dist_scale}. Since $c(x)$ takes $O(1)$ values throughout $x\in (0, \infty)$, its specific form does not change any scaling behavior of our interest. Hence, by setting $c(x)=1$, we obtain:
\begin{equation}
S_{A}= b_1(a) |A| + b_2(\eta) |A|^{2-\mu} + o(a, |A|^0),
\end{equation}
with
\begin{equation}
\begin{aligned}
    \eta &= |A|/t^{\nu},\\
    b_1(a) &= 1+a^{1-\mu}/(1-\mu),\\
    b_2(\eta) 
  &=\eta^{\mu-1}\Gamma(1-\mu, 0, \eta) + \eta^{\mu-2}
   \Gamma(2 - \mu,0,\eta),
\end{aligned}
\end{equation}
where $\Gamma(z,x_1,x_2)=\int_{x_1}^{x_2} s^{z-1}e^{-s}\d s$ is the generalized incomplete gamma function. The key observation is that the time-dependent part $b_{2}(\eta) |A|^{2-\mu}$ depends on $t$ and $A$ only through the ratio $\eta= |A|/t^{\nu}$.

\subsection{Multipartite entanglement generation}\label{sec:CMI}

We have observed that the stabilizer size distribution grows super-linearly and leaves some imprints on the sub-leading behavior of $S_{A}$ for a contiguous subsystem $A$. These contributions, however, are often hidden in the sub-leading terms, and are not manifest in simple bipartite entanglement measures. In this section, we probe the super-linear stabilizer growth using multipartite entanglement measures. 

The underlying difficulty in detecting the super-linear stabilizer growth is the absence of bi-partite entanglement in a weakly monitored circuit. Namely, the mutual information $I_{A:B}$ for two distant subsystems $A,B$ remains almost zero at any time. Here, we begin by presenting a derivation of a universal expression of the mutual information from the stabilizer length distribution. 

Let $A$, $C$ and $B$ be three consecutive intervals, with $C$ separating two disjoint intervals $A$ and $B$. Recall that the mutual information $I_{A:B}$ is given by 
\begin{align}
I_{A:B}=S_{A} + S_{B} - S_{AB}.
\end{align}
Here, $S_{A}, S_B$ can be computed from the stabilizer length distributions as $A,B$ are single contiguous regions. However, $S_{AB}$ is the entanglement entropy for a union of two disjoint intervals and cannot be readily computed. [Note that \refeq{eq:clipMI} considers the case where $A,B$ are neighboring.]

In order to overcome this difficulty, we shall utilize a stabilizer counting argument combined with a certain probabilistic assumption about the Pauli operator contents of stabilizers~\cite{yoshida2021decoding}. 
Let $\rho(t)$ be an evolving stabilizer state in a monitored circuit, with $G$ being $\rho(t)$'s stabilizer group. 
In the stabilizer formalism, the mutual information between $A$ and $B$ is given by:
\begin{equation}
    I_{A:B} =  \log |G_{AB}| - \log |G_{A}| - \log |G_{B}|,
\end{equation}
where $G_X$ is the subgroup of $G$ that only acts non-trivially on the subsystem $X$.

We now estimate $\log |G_{AB}|$, by assuming that each stabilizer in $G_{ABC}$ has random content within its support. Requiring stabilizers to be identity on $C$ would impose $2|C|$ independent constraints. However, since all stabilizers must commute with each other, in particular with elements in $G_C$, $\log |G_C|$ constraints are automatically satisfied. Hence, we arrive at the following estimation:
\begin{equation}
\begin{aligned}
        \log |G_{AB}|&\approx \log |G_{ABC}|-2|C|+\log |G_C|\\
        &=|AB|-S_{ABC}-S_C.
\end{aligned}
\end{equation}

The estimation combined with exact formulas for $\log |G_{A}|$ and $\log |G_{B}|$ gives rise to an unphysical negative number for $I_{A:B}$ when stabilizers with empty support on $C$ are rare. Taking this case into consideration, we arrive at the following expression for $I_{A:B}$:
\begin{equation}\label{eq:mutualinfoapp_0}
I_{A:B}\approx\max\{S_{A}+S_{B}-S_{ABC}-S_C, 0\}.
\end{equation}
In \refsec{sec:dwIAB} we will present another derivation of it based on the domain-wall picture.

The relation \refeq{eq:mutualinfoapp_0} holds regardless of $C$'s size.
In the case of $|C|>|A|,|B|$, recalling that the single interval entropy $S_{X}$ has a leading term that is linear in $|X|$ at any time, we conclude the first term in $\max\{\cdot\}$ is always negative and $I_{A:B}\approx 0$ at any time.
% It is worth emphasizing that  
% Later in \refsec{sec:coded}, we will find another usage of it when analyzing the coding properties of monitored dynamics. 

We now return to the discussions of multipartite entanglement measures and their relation to the sizes of stabilizer generators. The fact that $I_{A:B}\approx 0$ prompts us to consider multipartite entanglement measures. A particularly useful choice is the tri-partite information~\cite{hayden2013holographic}:
\begin{equation}
\begin{aligned}
    I_{A:B:C} 
    &\equiv S_A + S_B + S_C - S_{AB}-S_{AC}-S_{BC}+S_{ABC},
\end{aligned}
\end{equation}
which is also known as the topological entanglement entropy in studies of topological phases of matter~\cite{kitaev2006topological}. Roughly, the negativity of $I_{A:B:C}$ implies the presence of quantum entanglement among four subsystems $A,B,C$ and their complement $(ABC)^c$. See~\cite{hayden2007black, hosur2016chaos} for detailed discussions on various properties of $I_{A:B:C}$.

It turns out that, in a weakly monitored circuit, the tri-partite information is equivalent to the conditional mutual information:
\begin{equation}\label{eq:cmi}
    \begin{aligned}
  I_{A:B:C} = I_{A:B}- I_{A:B|C}\approx -I_{A:B|C}
    \end{aligned}
\end{equation}
since $I_{A:B}\approx 0$. 
This enables us to evaluate $I_{A:B:C}$ using the stabilizer length distribution via \refeq{eq:clipCMI}.

Let us consider two regions $A$ and $B$ separated by a long interval $C$. For concreteness, we let $|A|=|B|=r|C|$ for some small $r\ll1$, so that $|C|$ is the only length scale in the problem.
Recalling that $I_{A:B|C}$ exactly measures the number of stabilizers that connect $A$ and $B$, and the longest stabilizer at time $t$ has a size $\ell^*(t)$, we conclude $I_{A:B|C}(t)$ should go from zero to non-zero when $\ell^*\sim |C|\ \Leftrightarrow t\propto |C|^{2/3}$.
% When the cutoff length scale $\ell^*$ becomes comparable to $|C|$, entanglement entropies in contiguous subsystems can be written universally with the volume-law term and the sub-leading term. When computing $I_{A: B|C}$, the volume-law terms cancel with each other, and only the sub-leading terms survive. As such, we obtain $- I_{A:B:C} =  I_{A:B|C} \sim |C|^{2-\mu}$ for $t \gg |C|^{1/\nu}$.

An alternative and more quantitative way is the following. Using \refeq{eq:vN}, we know that $I_{A:B|C}$ must take a form
\begin{align}\label{eq:cmiscaling}
    I_{A:B|C}(t)=g(|C|/t^\nu)|C|^{2-\mu}.
\end{align}
This implies that the time-scale $t^*$ for $I_{A:B:C}$ to develop must be proportional to $|C|^{1/\nu}=|C|^{2/3}$, \textit{i.e.} $t^*$ grows sub-linearly with the $A$ and $B$'s separation. 

Both analyses lead to the conclusion that: $A$ and $B$ start to get entangled in a time scale that is sublinear in their distance $|C|$. We verify this statement numerically by simulating $I_{A: B|C}$ explicitly, as displayed in \reffig{fig:cmi}(a). The \reffig{fig:cmi}(b) shows the collapse of the simulated $I_{A:B|C}$ according to the scaling form \refeq{eq:cmiscaling}, showing the validity of the scale form \refeq{eq:cmiscaling}.
\begin{figure}[h]
    \centering
    \subfloat[]{
    \stackinset{r}{0.48\linewidth}{b}{0.30\linewidth}
    {\includegraphics[width=0.10\textwidth]{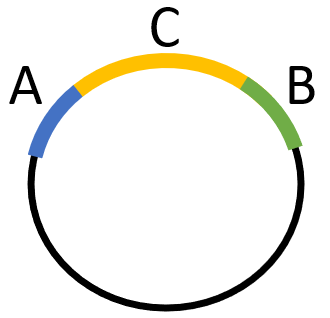}}
    {\includegraphics[width=0.40\textwidth]{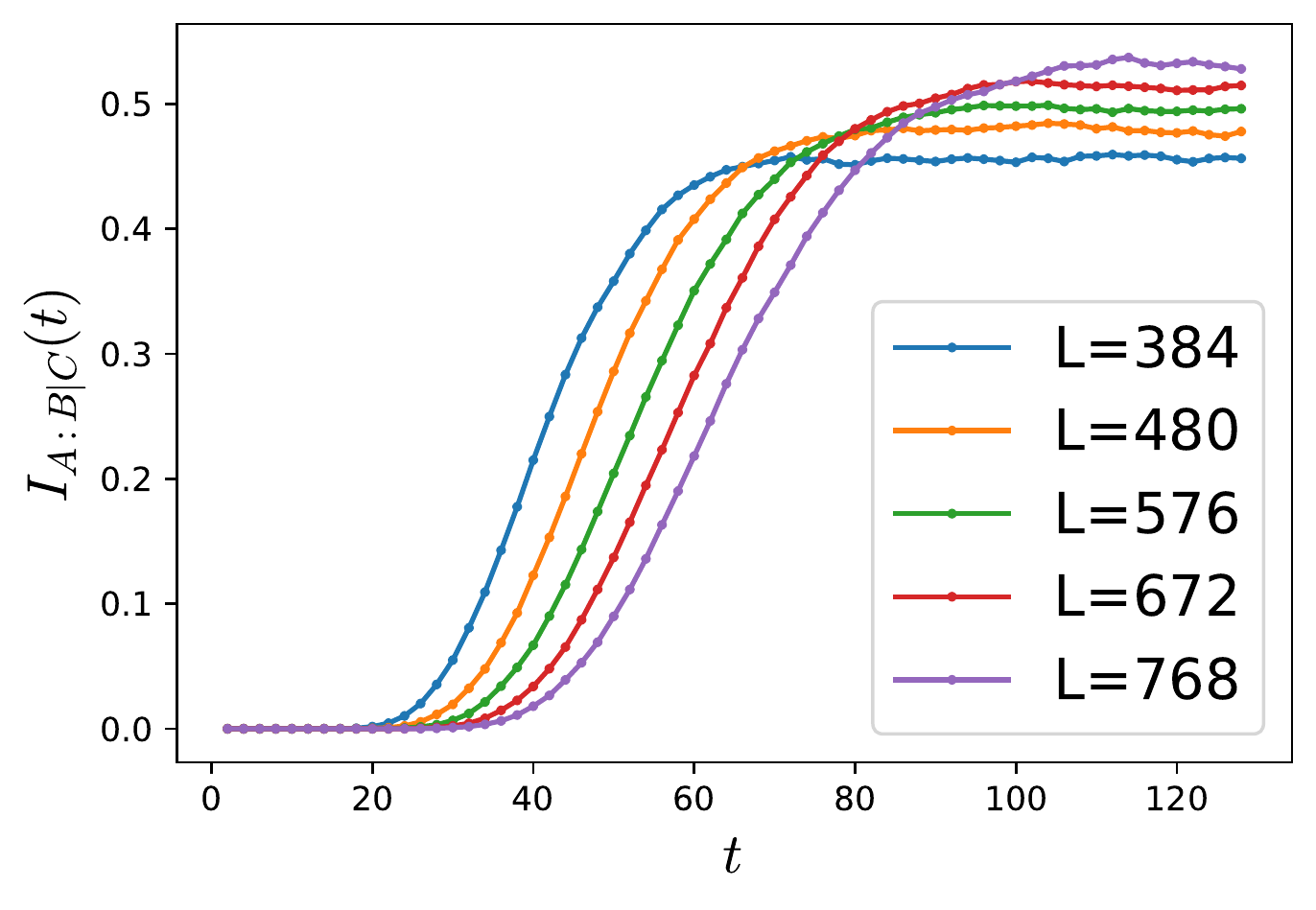}}
    }\\
    \subfloat[]{\includegraphics[width=0.40\textwidth]{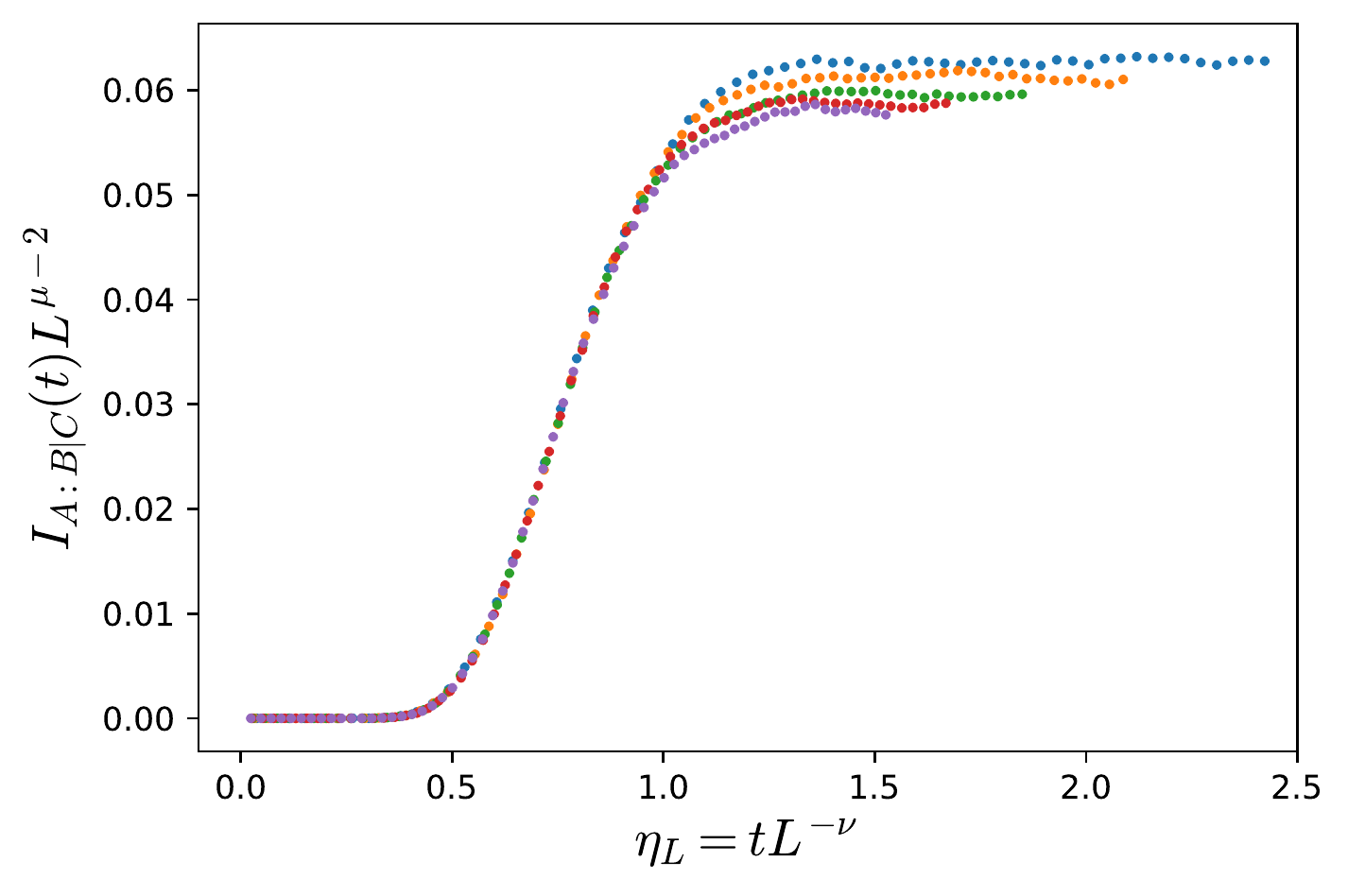}}
    \caption{(a) Dynamics of conditional mutual information $I_{A:B|C}(t)$, on a periodic spin chain for various $L$. Inset shows the geometry of $A$, $B$ and $C$, whose sizes are $|A|=|B|=\frac{5}{24}|C|=\frac{1}{8}L$ (b) The collapse of $I_{A:B|C}$'s data using the scaling
    form \refeq{eq:cmiscaling}, with $(\mu, \nu) = (1.66, 1.5)$. }
    \label{fig:cmi} 
\end{figure}

\subsection{Code distance growth}\label{sec:coded}

In this subsection, we study the entanglement dynamics from the perspective of quantum error-correcting codes. Namely, we argue that the sublinear time scale $t\sim L^{\frac{2}{3}}$ corresponds to the encoding time of a monitored quantum circuit. 

For a stabilizer code $\rho$, the number of independent logical operators in a region $A$ equals $I_{A:R}(\ket{\rho})$, where $\ket{\rho}$ is a purification of $\rho$ and $R$ is the reference system \cite{li2021statistical}. Thus, the code distance $d_{\mathrm{code}}$ can be expressed as:
\begin{equation}
    d_{\text{code}} = \min \{|A|:I_{A:R}(\ket{\rho})>0\}.
\end{equation}
In fact, $I_{A:R}(\ket{\rho})$ can be computed from $\rho$ directly. Denote $A^c$ as the complement of $A$ in the physical system and notice that $\ket{\rho}$ is a pure state on $AA^cR$. We then have:
\begin{equation}
    I_{A:R}(\ket{\rho}) = 2S_A(\rho) - I_{A:A^c}(\rho), 
\end{equation}
where the r.h.s. can be determined by the entanglement structure of $\rho$.

Returning to monitored quantum circuits, recall that the evolving state $\rho(t)$ can be viewed as a stabilizer error-correcting code whose dynamical code space is the subspace where $\rho(t)$ is supported \cite{gullans2020dynamical, li2021statistical}. The (contiguous) code distance $d_{\mathrm{code}}$ corresponds to the size of a minimal contiguous subsystem that supports a logical operator.

Let us begin with the equilibrium case with $t\gg L^{\frac{2}{3}}$. It has been shown that $d_{\mathrm{code}}$ is closely related to the subleading term of the steady state entanglement entropy in \refeq{eq:SsmallA}. Namely, if $d_{\mathrm{code}}\propto L^{\gamma_{\mathrm{code}}}$, then the entanglement entropy $S_{A}$ must have a sub-leading term with $|A|^{\gamma_{\mathrm{code}}}$~\cite{yoshida2021decoding}. Using this observation, we obtain the following relation of scaling exponents:
\begin{align}
 \gamma_{\mathrm{code}} = 2-\mu \approx \frac{1}{3},
\end{align}
which relates the power-law tail of the stabilizer distribution and the code distance. This is consistent with the numerical estimate of $\gamma_{\mathrm{code}}$ from~\cite{li2021statistical, li2021dpre}.

Next, let us consider the intermediate case with $t \ll L^{\frac{2}{3}}$.
Here, we find that the code distance grows as $d_{\mathrm{code}}\sim t^{1/2}$ before it reaches the equilibrium value of $\sim L^{1/3}$ at a time scale $t \sim L^{2/3}$. This can be derived by following an argument from~\cite{yoshida2021decoding}. Namely, we pick two equal-sized contiguous regions $B_1$ and $B_2$ to the left and right of $A$. When $|B_1|=|B_2|$ is sufficiently large, we have $I_{A: A^c}=I_{A:B}$ where $B=B_1\cup B_2$. 
By making use of the relation \refeq{eq:mutualinfoapp_0}, we obtain:
\begin{equation}\label{eq:IAB}
\begin{aligned}
2S_A-I_{A:B}
&=S_{AB}+S_A-S_{B_1}-S_{B_2}+I_{B_1:B_2}\\
&\approx\max\{0,S_{AB}+S_A-S_{B_1}-S_{B_2}\}.
\end{aligned}
\end{equation}
The second term is negative when $|A|$ is small, and it gradually switches to positive values as $|A|$ increases. In this regime, we also have $I_{B_1:B_2}\approx 0$. Thus, the code distance is given by the size of $A$ that satisfies:
\begin{equation}
    S_{B_1}+S_{B_2}-S_{AB}-S_A \approx 0.
\end{equation}
Recalling that $B_1$, $B_2$, and $A\cup B$ are contiguous regions, we can compute their entanglement entropies using \refeq{eq:SlargeA2}. When $|A|\ll t^{3/2}$, we find
\begin{equation}
\begin{aligned}
    &2\left((b_1 + c_1 t^{-1})|B_1|+c_2 t^{1/2}\right)\\
    &\approx 
    (b_1 + c_1 t^{-1})(2|B_1|+|A|)+c_2 t^{1/2} + b_1|A| + b_2 |A|^{1/3}\\
    &\Rightarrow 
    \ d_{\text{code}} =|A|\ \approx \frac{c_2 t^{1/2}}{2b_1+c_1 t^{-1}}\sim t^{1/2}.
\end{aligned}
\end{equation}
We confirmed this scaling numerically, by simulating the following quantity:
\begin{equation}\label{eq:epsilon_code_dist}
    d_{\text{code},\epsilon}(t)\equiv \text{argmax}_{|A|} (I_{R:A}(t)\leq\epsilon).
\end{equation} 

\begin{figure}[h]
    \centering
    \includegraphics[width=0.40\textwidth]{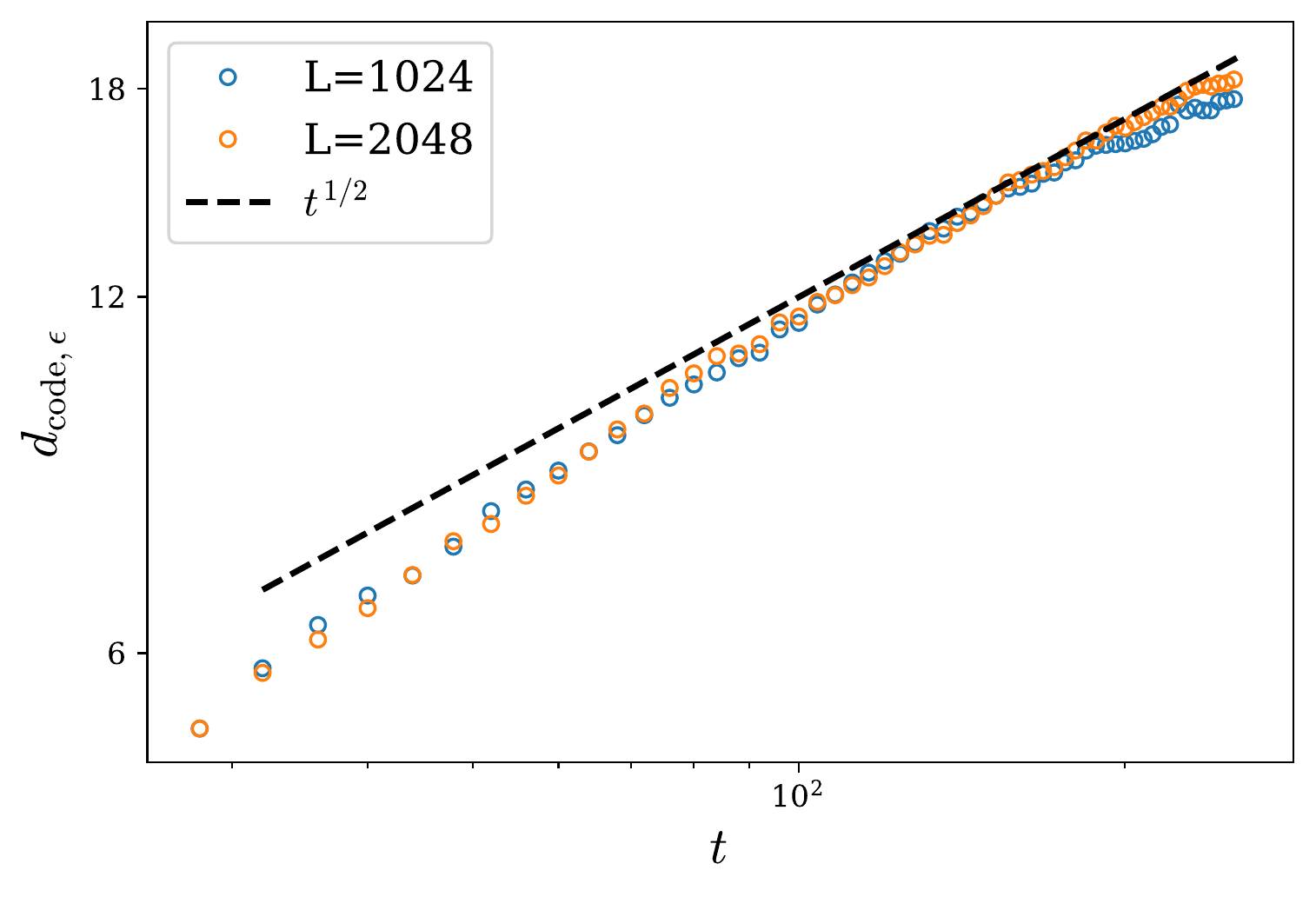}
    \caption{Contiguous code distance, as defined in \refeq{eq:epsilon_code_dist}, plotted in the log-log scale. $\epsilon$ is taken to be $1$.}
    \label{fig:code_dist} 
\end{figure}

\subsection{Code equilibrium v.s. entanglement equilibrium}
% \tim{Move to after Section F?  Requires knowledge of pure initial states.} 
Superlinear growth of stabilizer generators leads to a new sublinear time scale $t \simeq L^{\frac{2}{3}}$ which can be identified as an encoding time of a dynamically generated stabilizer code. Indeed, we have seen that the stabilizer length distribution, as well as the contiguous code distance, reaches a steady form after $t \simeq L^{\frac{2}{3}}$. This suggests that the resulting output state $\rho(t)$ is close to a maximally mixed state inside the dynamically generated code subspace after the code equilibrium time of $t \simeq L^{\frac{2}{3}}$. 

This observation, however, appears to be in tension with previous results on monitored quantum circuits. Namely, if one starts the circuit from states with low entanglement (such as product states), the evolving state will acquire a steady form of entanglement structure and reach the volume-law entanglement only at a much later time of $t \simeq L$ \cite{skinner2019measurement}. Some readers might wonder if this suggests that the encoding time would actually be $t \simeq L$ instead of $t \simeq L^{\frac{2}{3}}$. 

The key to resolving this apparent puzzle is to observe that not all the states in the code subspace possess volume-law entanglement. 
To understand this point, let us consider the entanglement entropy $S_{A}(t)$ for different initial states at the encoding time $t \simeq L^{\frac{2}{3}}$. 
When the initial state of the circuit is a product state, the output state $|\psi_{\mathrm{prod}(t)}\rangle$ does not possess a volume-law entanglement. Namely, we have $S_{A}(t) \simeq |A|$ only up to $|A|\lessapprox L^{\frac{2}{3}}$, and $S_{A}(t)\simeq L^{\frac{2}{3}}$ for $|A|\gtrapprox L^{\frac{2}{3}}$. 
These output states $|\psi_{\mathrm{prod}(t)}\rangle$, however, are still in the code subspace. 
Next, when the initial state of the circuit is a Haar random state, the output state $|\psi_{\mathrm{Haar}(t)}\rangle$ has a volume-law entanglement. 
Namely, one can show that $S_{A}$ for $|\psi_{\mathrm{Haar}(t)}\rangle$ behaves in a similar way to the one for the case with a maximally mixed initial state as long as $|A|\leq \frac{L}{2}$~\cite{yoshida2021decoding}. 
Here, $|\psi_{\mathrm{Haar}(t)}\rangle$ corresponds to a typical random state inside the code subspace. 

The aforementioned observation suggests that, while there exist atypical codeword states $|\psi_{\mathrm{prod}(t)}\rangle$ with $\simeq L^{\frac{2}{3}}$ entanglement, typical codeword states indeed possess the volume-law entanglement. 
This resolves the puzzle concerning $t \simeq L^{\frac{2}{3}}$ v.s. $t \simeq L$; the former is the required time to reach the codeword state while the latter is the time to attain entanglement properties of typical codeword states.

\subsection{Pure v.s. mixed initial state}\label{sec:pure_vs_mix}

So far in this section, we have focused on cases where the initial state of the circuit is the maximally mixed state. A naturally arising question concerns what will happen if we start from a pure product state instead. It has been previously noted that the two choices lead to strikingly different behaviors of the entanglement dynamics at late time \cite{li2019measurement, gullans2020dynamical}. Here we revisit this problem, focusing on whether the super-linear growth is also present in the pure initial state case.

\begin{figure}[h]
% \captionsetup[subfigure]{aboveskip=-2pt,belowskip=-2pt}
\centering
\subfloat[]{\includegraphics[width=0.40\textwidth]{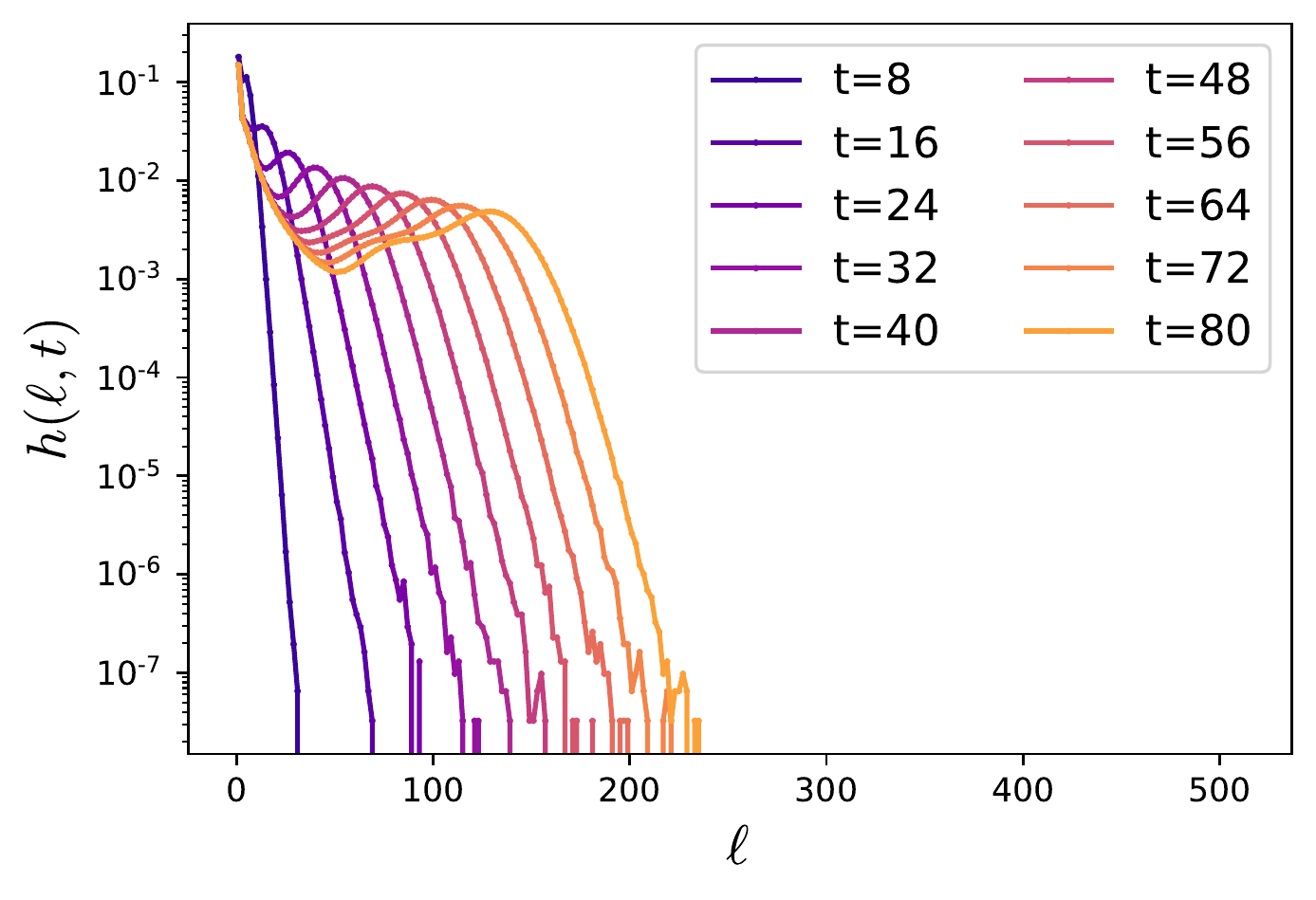}}\\
\subfloat[]{
\stackinset{r}{0.45\linewidth}{b}{0.38\linewidth}
{\includegraphics[width=0.13\textwidth]{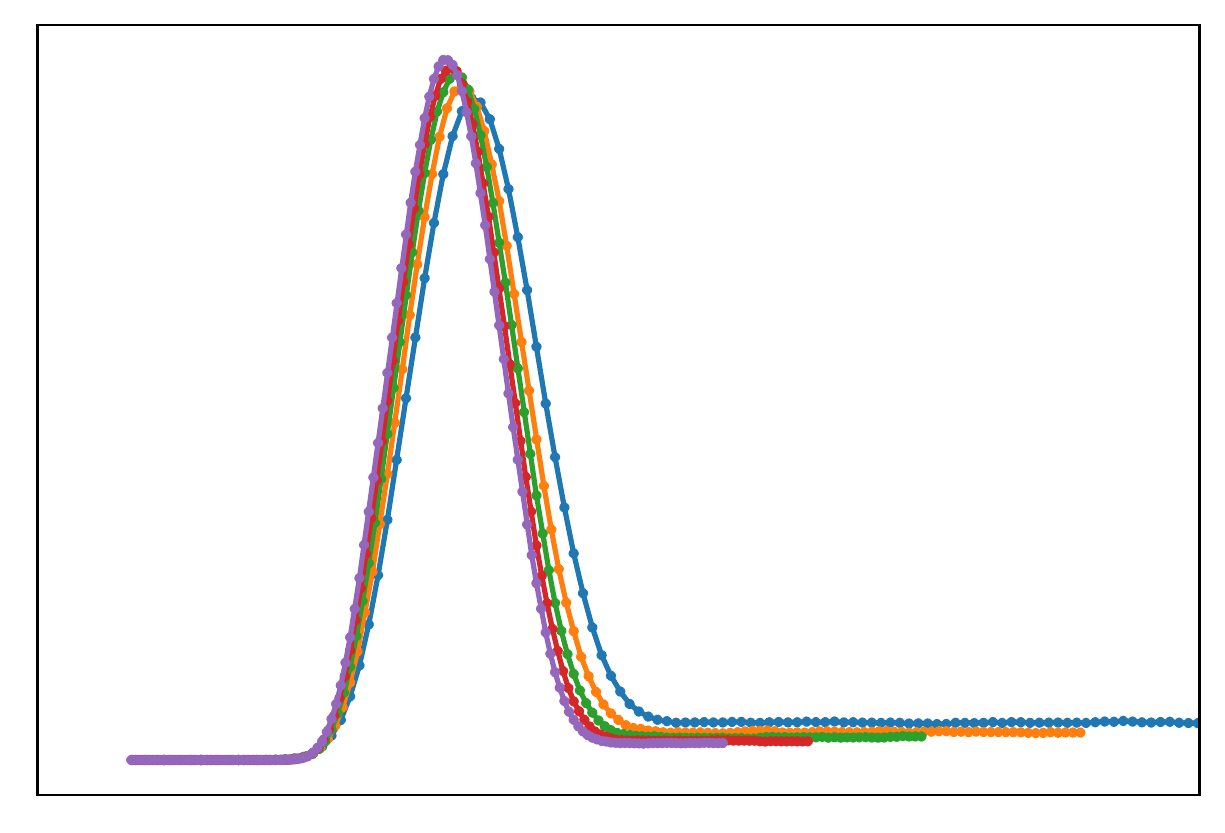}}
{\includegraphics[width=0.40\textwidth]{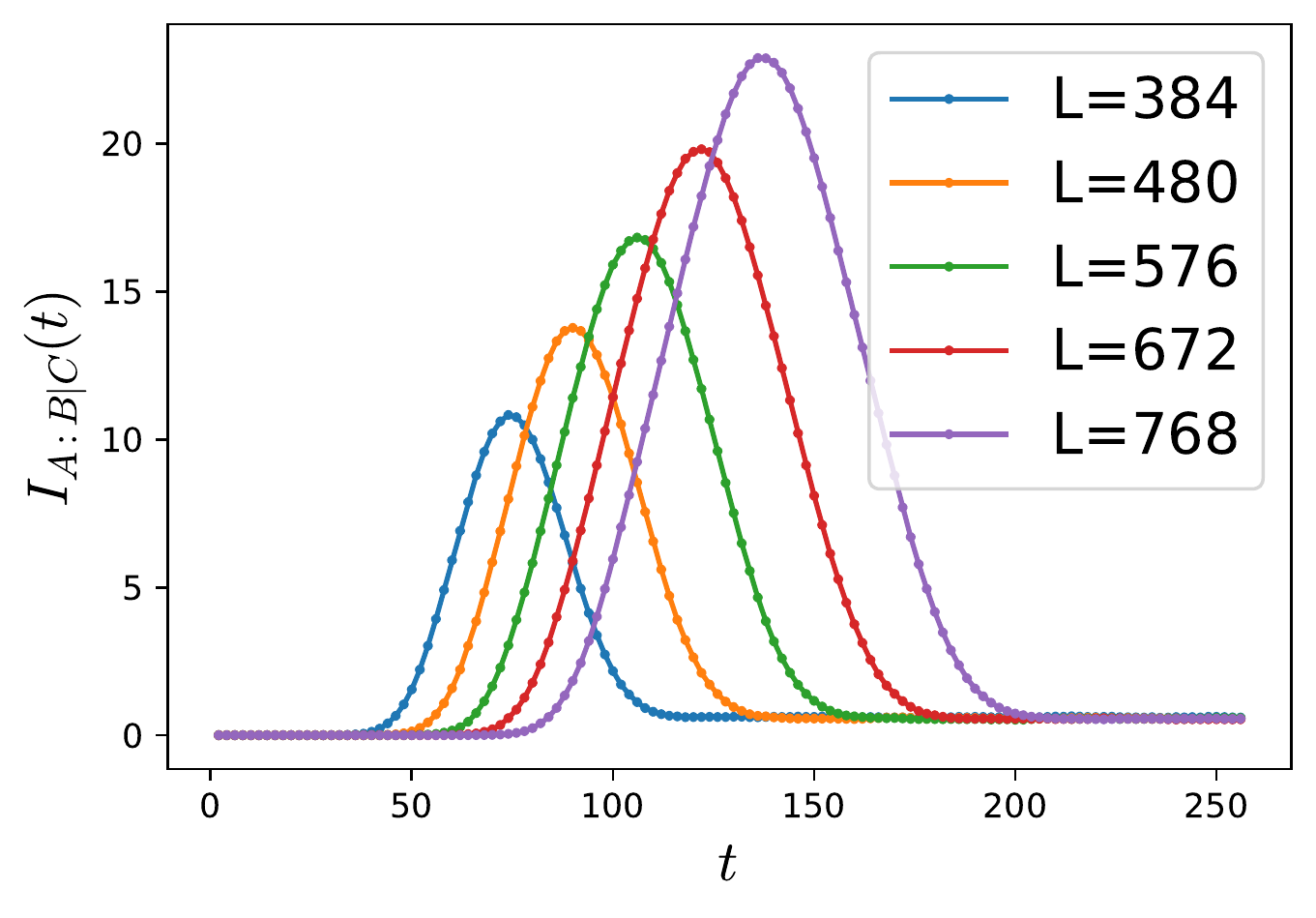}}
}
\caption{Numerical results for a pure initial state. (a) Stabilizer length distribution $h(l,t)$ with $L=512$ and open boundary condition, within the same time period as in \reffig{fig:stab_dist} (b) Dynamics of $I_{A:B|C}(t)$, with $A$, $B$ and $C$ arranged in the same manner as in \reffig{fig:cmi}(a). The inset shows the same data but with both axes rescaled by a factor of $L^{-1}$.}
\label{fig:pure_initial_state}
\end{figure}

We present the simulated $h(\ell, t)$ for a pure product initial state in the \reffig{fig:pure_initial_state}(a). Recall that, for a pure product state, all stabilizer generators are of size one, and thus, the initial distribution is a sharp peak $h(\ell, t=0)\propto\delta(\ell-1)$. For $t>0$, the short stabilizer generators will grow due to the action of unitary gates, which corresponds to a right-moving peak. On the left side of the peak, one finds a power-law decaying profile that remains unchanged even after the peak passes. Here, the exponent $\frac{5}{3}$ is the same as the one observed in the mixed state case. This suggests that the power-law profile in $h(\ell,t)$ is present regardless of the choice of initial states. On the right side of the peak, $h(\ell, t)$ exponentially decays to zero, with an exponent that does not seem to be strongly dependent on time. 

Importantly, numerical results suggest that the pure evolving state does not possess any super-linear growth. Indeed, the peak only moves linearly with $t$, and $h(\ell,t)$ decays quickly after the peak. This is also reflected in the simulation of $I_{A: B|C}(t)$ [\reffig{fig:pure_initial_state}(b)]. We find that the time scale for $I_{A:B|C}$ to take non-zero value scales linearly with $|C|$, in contrast to the $t \sim |C|^{\frac{3}{2}}$ in the mixed state case. 

To gain some understanding of entanglement dynamics for this case, we propose the following scaling form of $h(\ell, t)$ which captures the salient features described above:
\begin{equation}
    h(\ell,t) = c'_1 \frac{\Theta(\ell^*(t)-\ell)}{\ell^{5/3}} + c'_2 \delta(\ell-\ell^*(t)),
\end{equation}
where $\ell^*(t)\approx t$ is the position of the peak. Though the peak in the figure does get broader with time, we approximate it with a delta function for ease of discussion. 

By making use of \refeq{eq:clipS} and the initial condition $S_A(0)=0$, we derive that a contiguous region's entanglement entropy has the following form:
\begin{equation}
    S_A(t) = 
    \left\{
    \begin{array}{lll}
         b_1 t + b_2 t^{\frac{1}{3}} & &\ell^*(t)<|A| \\
         b_1|A| + b_2|A|^{\frac{1}{3}} & &\ell^*(t)>|A|
    \end{array}
    \right.
\end{equation}
where the $b_{1}$ and $b_2$ are the same as those appeared in \refeq{eq:SsmallA}.
This entanglement growth is similar to the one in the pure unitary dynamics starting from a pure state, where $S_A$ typically grows linearly until saturation. 

We can also use the expression above to compute $I_{A:B|C}(t)$. By keeping the leading term with $b_{1}$ only, we have
\begin{equation}
\begin{aligned}
    &I_{A:B|C}(t)\\
    =& (S_{AC} + S_{BC} - S_{ABC} - S_{C})(t)\\
    =&
        \left\{
        \begin{array}{llc}
             0 & &\frac{\ell^*(t)}{|C|}<1 \\
             b_1(t-|C|) & &1<\frac{\ell^*(t)}{|C|}<1+r\\
             b_1((1+2r)|C|-t) &  &1+r<\frac{\ell^*(t)}{|C|}<1+2r\\
             0 & &\frac{\ell^*(t)}{|C|}>1+2r
        \end{array}
        \right.
\end{aligned}
\end{equation}
which indeed captures salient features of the simulated $I_{A:B|C}(t)$ in \reffig{fig:pure_initial_state}(b).

\section{Information dynamics}\label{sec:info-spread}

In the previous section, we studied how the entanglement structure of the output state $\rho_{\mathrm{out}}(t)$ evolves. In this section, we consider how local information about the input state spreads out by studying correlations between the input and the output states. 

\subsection{Input-Output correlation}

In order to study information spreading, it is convenient to introduce a reference system $R=\{1_R, 2_R, ..., L_R\}$ for the original physical system $P=\{1_P, 2_P, ..., L_P\}$. Each qubit in $R$ is initially maximally entangled with its partner in $P$, and the monitored circuit acts on the physical system only, see \reffig{fig:two_copy} for an illustration. The resulting state is
\begin{equation}\label{eq:vectorize}
\ket{\phi_{\textbf{m}}(t)}\propto \sum_{\textbf{s}\in \{0,1\}^L}(C_{\textbf{m}}(t)\ket{\textbf{s}})_P\otimes \ket{\textbf{s}}_R
\end{equation}
where $C_{\textbf{m}}$ is defined in a way similar to \refeq{eq:evolution} and $\textbf{m}$ represents the measurement outcome. The output state $\rho_{\mathrm{out}}(t)$ of a monitored circuit, when starting from a maximally mixed state, can be obtained by tracing out the reference system $R$ as $\rho_{\mathrm{out}}(t) = \mathrm{Tr}_{R} \big(\ket{\phi_{\textbf{m}}(t)} \bra{\phi_{\textbf{m}}(t)}\big)$.

\begin{figure}[h]
\centering
\includegraphics[width=0.3\textwidth]{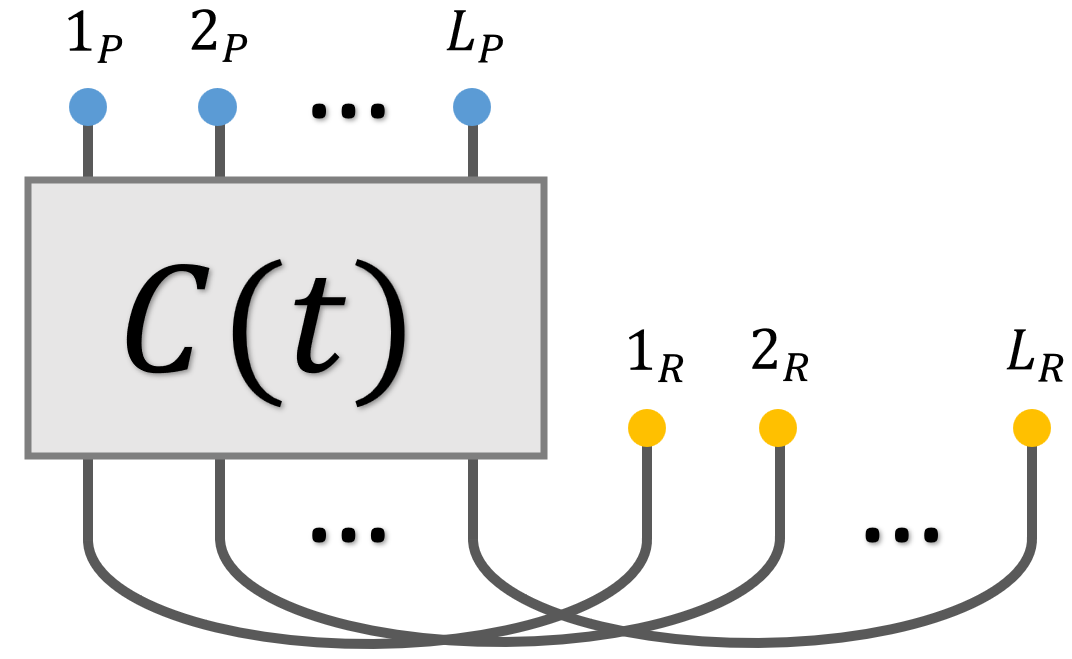}
\caption{
A visualization of $\ket{\phi_{\textbf{m}}(t)}$ defined in \refeq{eq:vectorize}. Physical qubits ($P$, blue dots) and reference qubits ($R$, yellow dots) are initially maximally entangled and the circuit acts on the physical qubits only.
}
\label{fig:two_copy}
\end{figure}

Under the action of monitored circuits, some of the information initially contained in a subregion $A_P$ will gradually spread without being eliminated by measurements. 
To quantify this effect, we utilize the mutual information between the output and the input spaces:
\begin{equation}\label{eq:mi}
I_{A_R: B_P}(t) = (S_{A_R} + S_{B_P} - S_{A_R\cup B_P})(\ket{\phi_{\textbf{m}}(t)})
\end{equation}
for various subregions $B_P\subseteq P$. Intuitively, this quantity measures how much information is transmitted from $A$ to $B$ by the monitored circuit $C_{\textbf{m}}(t)$. 

In the remainder of this section, we will provide a physical interpretation of $I_{A_R: B_P}(t)$ and its relation to the notion of operator spreading. 
We then study the behavior of $I_{A_R: B_P}(t)$ in both monitored and unitary $(1+1)$-dimensional circuits.

\subsection{Isomorphism between quantum codes}\label{sec:Troperators}

A single trajectory of Clifford monitored quantum dynamics $C_{\textbf{m}}$ can be viewed as an isomorphism between two quantum error-correcting codes. To see this, we perform the singular value decomposition to $C_{\textbf{m}}$:
\begin{equation}\label{eq:Csvd}
    C_{\textbf{m}} = \sum_{i=1}^K \ket{v_i}\bra{u_i}.
\end{equation}
Note that all the non-zero singular values of $C_{\textbf{m}}$ are $1$ as $C_{\textbf{m}}$ is a Clifford monitored dynamics. The decomposition defines two stabilizer code subspaces: $\hilbert_{\text{in}}=\langle u_1,..., u_K \rangle$ and $\hilbert_{\text{out}}=\langle v_1,..., v_K \rangle$, and corresponding projectors are given by
\begin{align}
P_{\text{in}}=C_{\textbf{m}}^\dagger C_{\textbf{m}}, \quad P_{\text{out}}=C_{\textbf{m}}C_{\textbf{m}}^\dagger.
\end{align}
Furthermore, the input and output stabilizer code states $\rho_{\text{in}}, \rho_{\text{out}}$ are given by maximally mixed states of $|u_{i}\rangle$ and $|v_{i}\rangle$ respectively.

Using this decomposition, we can view the dynamics $C_{\textbf{m}}$ as a two-step process: for an incoming state $\ket{\psi}$, $C_{\textbf{m}}$ will first discard the part of $\ket{\psi}$ that is orthogonal to $\hilbert_{\text{in}}$:
\begin{equation}
    \ket{\psi}\Rightarrow P_{\text{in}}\ket{\psi},
\end{equation}
then unitarily transform the remaining wavefunction into the subspace $\hilbert_{\text{in}}$:
\begin{equation}
    P_{\text{in}}\ket{\psi} \Rightarrow C_{\textbf{m}} P_{\text{in}}\ket{\psi}.
\end{equation}
Hence, $\hilbert_{\text{in}}$ is the input subspace whose information is preserved, while $\hilbert_{\text{out}}$ is the subspace into which the preserved information is unitarily processed. 

By viewing the monitored dynamics $C_{\textbf{m}}$ as an isomorphism between two stabilizer codes, a correspondence between logical operations for $\hilbert_{\text{in}}$ and for $\hilbert_{\text{out}}$ can be naturally defined. 
Assume that $O$ and $\tilde{O}$ are non-trivial logical operators defined on $\hilbert_{\text{in}}$ and $\hilbert_{\text{out}}$ respectively. We say that $O$ is \textit{transferred} to $\tilde{O}$ by the dynamics $C_{\textbf{m}}$ if and only if:
\begin{equation}\label{eq:transfer_logical}
    C_{\textbf{m}} O = \tilde{O} C_{\textbf{m}},
\end{equation}
namely, applying $O$ on the input state is equivalent to applying $\tilde{O}$ on the output state. It is worth noting that this is not a one-to-one correspondence as logical operators have multiple equivalent expressions which act in the same manner in code subspaces.

The notion of transferred operators, defined above, is a natural generalization of operator spreading to non-unitary settings. 
Note that, in the special case that $C$ is a unitary evolution, both $\hilbert_{\text{in}}$ and $\hilbert_{\text{out}}$ are the total Hilbert spaces and any operator $O$ is a logical operator. 
In particular, its transferred operator $\tilde{O}$ is uniquely given by its unitary evolution under $C$ as $\tilde{O}=C^\dagger O C$.

The above discussions enable us to discuss information spreading in a quantitative manner by studying the growth of logical operators under an isomorphism $C_{\textbf{m}}$. Letting $A\subseteq P$ (input) and $B\subseteq P$ (output) be two subsets of qubits, one can ask how many inequivalent input logical operators defined on $A$, can be transferred to output logical operators defined on $B$. Namely, we are interested in the number of inequivalent $O_A$ such that:
\begin{equation}\label{eq:recoverlogical}
\exists \tilde{O}_B\text{ supported on $B$} \quad \tilde{O}_B C_{\textbf{m}}=C_{\textbf{m}}O_A.
\end{equation}

The following theorem relates this number to the mutual information $I_{A_R:B_P}$ defined in \refeq{eq:mi}:
\begin{theorem}\label{thm:transferable}
$I_{A_R:B_P}$ equals the number of independent and inequivalent Pauli logical operators supported on $A$ that can be transferred by $C_{\textbf{m}}$ to some logical operator within $B$. 
\end{theorem}
This theorem can be easily proven using \refeq{eq:clipS}, once we notice that the expression \refeq{eq:recoverlogical} is equivalent to $O_A^{T}\otimes \tilde{O}_B\ket{\phi_{\textbf{m}}}=\ket{\phi_{\textbf{m}}}$, namely $O_A^{T}\otimes \tilde{O}_B$ is a stabilizer of $\ket{\phi_{\textbf{m}}}$.

\add{
\subsection{Logical operator growth}\label{sec:reversal}
The study of operator growth plays a central role in our understanding of unitary circuits. Given the formalism developed in the previous subsection, it is natural to ask how a logical operator's size changes in time when it is transferred in a monitored circuit. In this part, we provide an argument showing that an $O(1)$ sized logical operator will be transferred to an $O(t^{3/2})$-sized one by $t$ layers of circuit dynamics.

% A key part of the argument is to consider the process of deepening the circuit $C$ by adding new operations to $C$, one at a time. We inspect how the output code state $\rho_{\text{out}}$ changes in the process. Here we choose to add new operations to the initial time slice $C$, which can simplify our discussion.

A key part of the argument is to compare output code states $\rho_{\text{out}}$ for circuits of different depths. 
To do so, we deepen the circuit by adding new operations to the initial time slice $C$, one at a time.

First, if  we add a unitary gate, then:
\begin{equation}
    C' = C U,
\end{equation}
which induces no change to the output code state: 
\begin{equation}
    \rho_{\text{out}}' \propto C'{C'}^\dagger \propto  CC^\dagger = \rho_{\text{out}}.
\end{equation}
Next, we consider the other case that a $Z_i$ basis measurement is added to $C$:
\begin{equation}
    C' = C\frac{1+Z_i}{2}.
\end{equation}
The impact on $\rho_{\text{out}}$ depends on whether $Z_i$ is a non-trivial logical operator of $\rho_{\text{in}}\propto C^\dagger C$. Namely, if $Z_i$ is not a logical operator of $\rho_{\text{in}}$, or if $Z_i$ is just a trivial logical operator of $\rho_{\text{in}}$, then one can show that it leaves the $\rho_{\text{out}}$ unchanged;
while if $Z_i$ is a non-trivial logical operator of $\rho_{\text{in}}$, it changes $\rho_{\text{out}}$ in the following way:
\begin{equation}\label{eq:changeout}
    \rho_{\text{out}}' = \frac{1+\tilde{Z}_i}{2} \rho_{\text{out}},
\end{equation}
where $\tilde{Z}_i$ is a non-trivial logical operator of $\rho_{\text{out}}$, which is $Z_i$'s transferred operator as defined in \refeq{eq:transfer_logical}:
\begin{equation}
    C Z_i = \tilde{Z}_i C.
\end{equation}
From the perspective of $\rho_{\text{out}}$'s stabilizers $\stabgen_{\text{out}}$, the change \refeq{eq:changeout} corresponds to adding a new stabilizer to $\stabgen_{\text{out}}$:
\begin{equation}
    \stabgen'_{\text{out}}  = \stabgen_{\text{out}}\cup\{\tilde{Z}_i \}.
\end{equation}

We can now infer the size of $\tilde{Z}_i$ by relating our observations above to our earlier conclusion that $\stabgen_{\text{out}}$'s stabilizer length distribution $h(\ell, t)$ follows \refeq{eq:len_dist}. Since the only part of $h(\ell, t)$ that changes at time $t$ is the small region around $\ell=\ell^*(t)\sim t^{3/2}$, we expect that $\stabgen_{\text{out}}$'s newly added stabilizer $\tilde{Z}_i$ is of a size $\sim t^{3/2}$. Thus $Z_i$ is transferred to a $\sim t^{3/2}$ sized operator $\tilde{Z}_i$ by $C(t)$. We will confirm this conjecture numerically in \refsec{sec:wavefront}. 

Here, we would like to note one caveat in the aforementioned argument concerning the size of $\tilde{Z}_j$. 
Recall that the size distribution in \refeq{eq:len_dist} is defined for stabilizer generators in a clipped gauge.
Adding $\tilde{Z}_j$ to stabilizer generators in a clipped gauge does not necessarily create stabilizers in a clipped gauge. 
As such, the evolution of \refeq{eq:len_dist} may be sourced by the addition of $\tilde{Z}_j$ as well as some rearrangement of stabilizer generators to fit in a clipped gauge. 
Despite this subtlety, we expect that such rearrangement will not significantly alter our argument. 
For one thing, in the time regime of our interest, $\mathcal{S}_{\text{in}}(t)$ is not full rank, so there is an $O(1)$ chance that we do not need any rearrangement to a clipped gauge at all. 
Furthermore, even when $\tilde{Z}_j$ breaks a clipped gauge, recombinations with other stabilizer generators do not change the size of $\tilde{Z}_j$ much. 
Hence, we expect that our estimate remains qualitatively valid. 
}

\subsection{Loss of local information}\label{sec:forget}

In this subsection and the next, we study information spreading through numerical simulations of $I_{A_R: B_P}(t)$. We perform simulation of $I_{A_R: B_P}(t)$ for both unitary ($p=0$) and monitored ($p=0.08$) cases. We set $A=[-a, a]$ and $B=[-d, d]$, with the origin taken to be the midpoint of the corresponding system ($P$ or $R$). An illustration of the setting is shown in \reffig{fig:spreading_AandB}. 
For the brevity of the presentation, we use the symbol $\mathcal{I}_A$ as a short-hand notation for the input information within $A$. As we have discussed, $I_{A_R: B_P}(t)$ measures the amount of $\mathcal{I}_A$ detectable within region $B\subseteq P$ at time $t$.

\begin{figure}[h]
    \centering
    \includegraphics[width=0.28\textwidth]{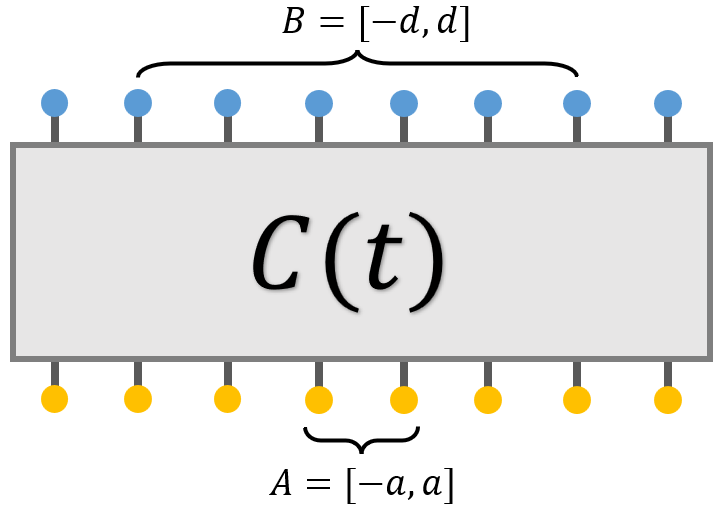}
     \caption{Arrangement of $A_R=[-a,a]_R$ and $B_P=[-d,d]_R$ in the numerical simulation. The figure shows the special case of $a=1$ and $d=3$.}
    \label{fig:spreading_AandB}
\end{figure}

We start by considering $I_{A_R:P}(t)=I_{A_R:(-\infty,\infty)_P}(t)$, \textit{i.e.} the amount of remaining $\mathcal{I}_A$ in the whole system at time $t$. In the unitary ($p=0$) case, we know:
\begin{equation}
    I_{A_R:P}(t)\equiv 2|A|,
\end{equation}
because unitary dynamics acts on $P$ only and  $R$ is always maximally entangled with $P$.  This implies that $\mathcal{I}_A$ is always preserved somewhere within the total system $P$, as is expected to be the case for a unitary dynamics.

When measurements are turned on ($p\neq 0$), they may destroy $\mathcal{I}_A$, as reflected in a gradual decay of $I_{A_R: P}(t)$, see \reffig{fig:decay}(a). It is natural to expect that $I_{A_R:P}(t)$ will eventually become zero and $\mathcal{I}_A$ will be completely lost, at some time scale that depends on $|A|=2a$. 

\begin{figure}[h]
    \centering
    \includegraphics[width=0.32\textwidth]{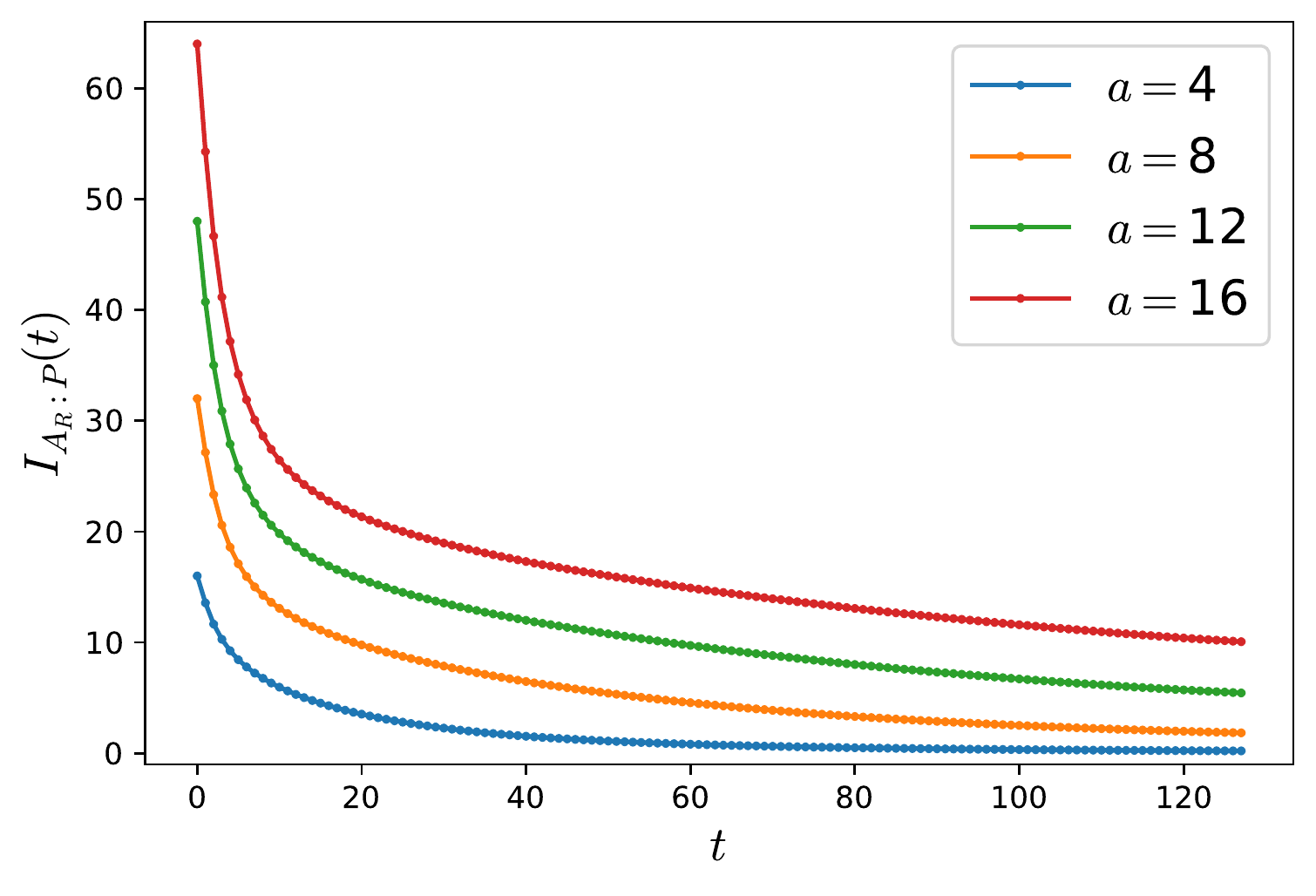}
     \caption{Decay of the remaining $\mathcal{I}_A$ within the system for various different $A$'s sizes.}
    \label{fig:decay}
\end{figure}

It turns out this problem, concerning the loss of local information, is equivalent to the code distance growth which we have studied in \refsec{sec:coded}. 
Monitored circuits respect time reversal symmetry on average: \add{ $C^\dagger(t)$, which is the reversed circuit of $C(t)$, has the same distribution as $C(t)$}.  This implies that 
\begin{equation}
    I_{A_R: P}(t) = I_{A_P: R}(t).
\end{equation}
Note that this equality is defined for averaged quantities, and may not hold for each instance.  
In \refsec{sec:coded}, we derived that the \textit{r.h.s.} vanishes when $|A|<O(\sqrt{t})$. This suggests that the time scale for $\mathcal{I}_A$ to be completely lost is given by
\begin{equation}\label{eq:timeloss}
    t^*_{\text{loss}} \sim |A|^2.
\end{equation}

\subsection{Information spreading}\label{sec:wavefront}
\begin{figure}
    \centering
    \subfloat[]{\includegraphics[width=0.45\textwidth]{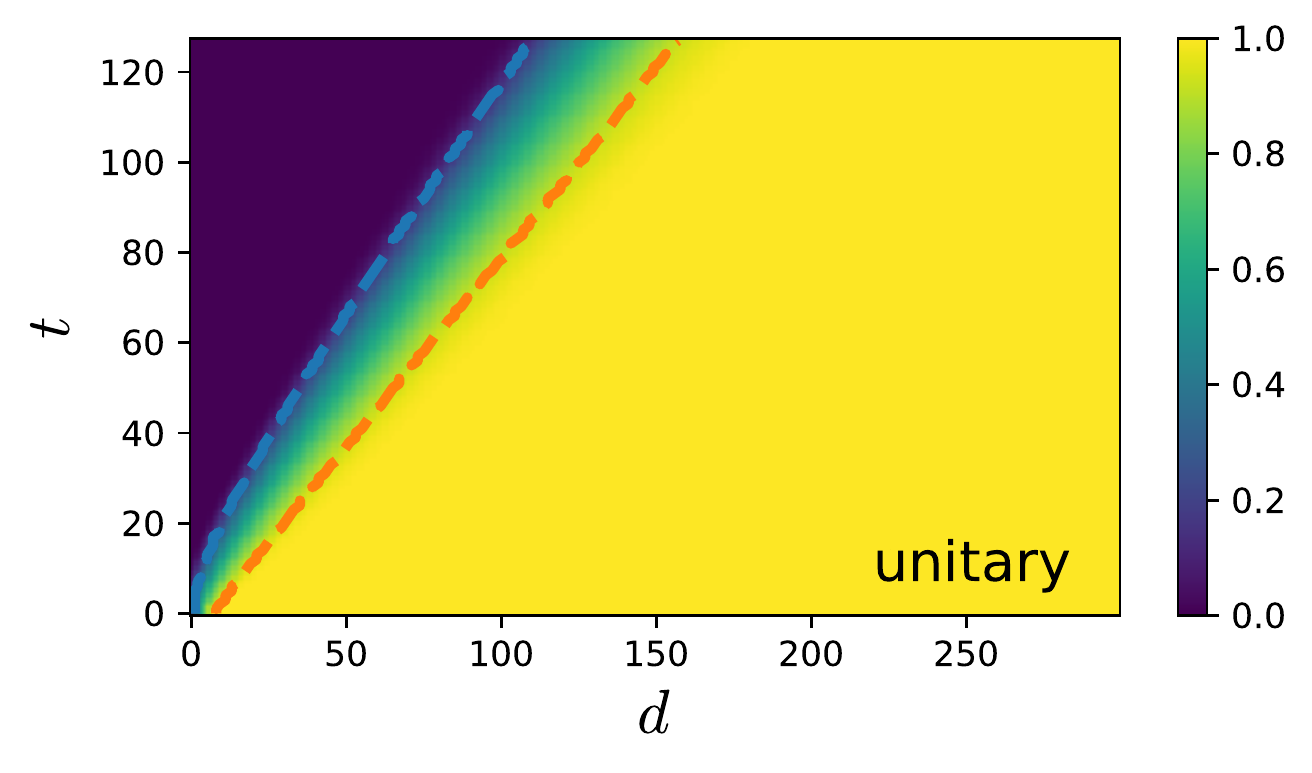}}\\
    \subfloat[]{\includegraphics[width=0.45\textwidth]{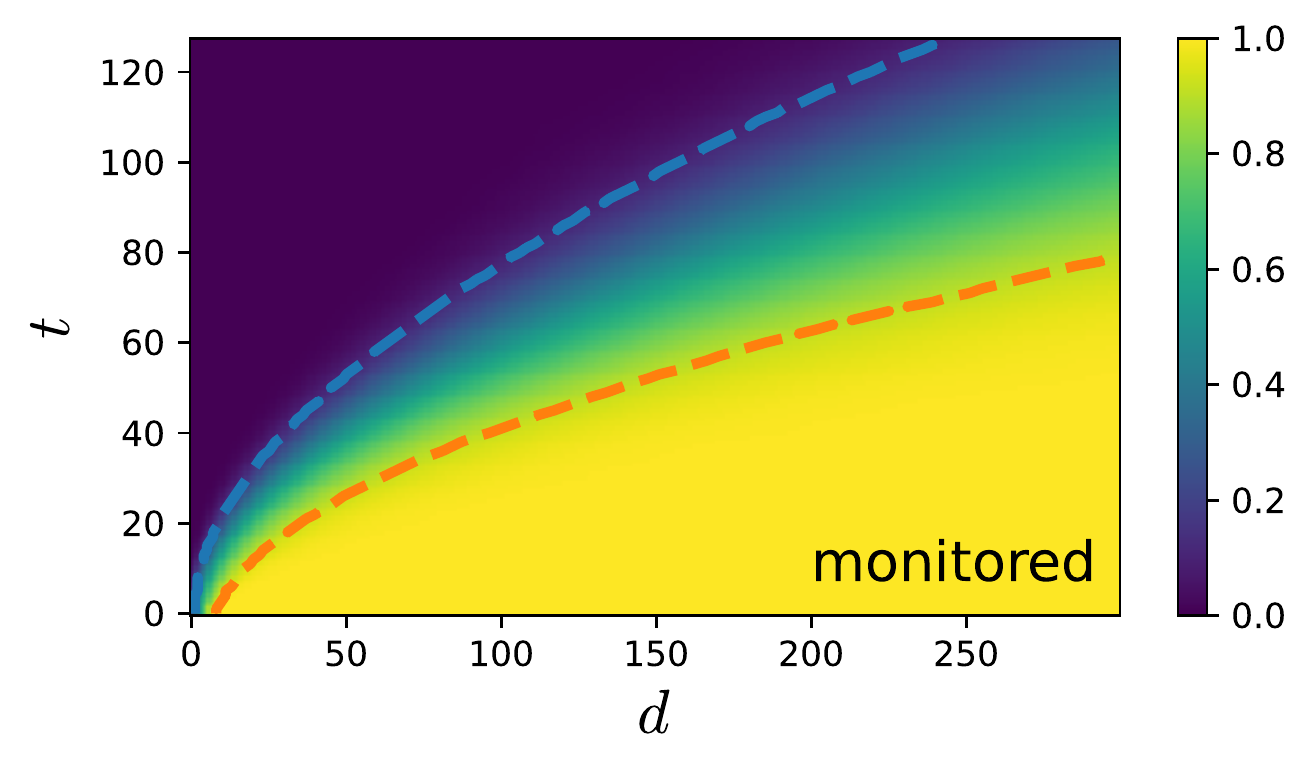}}\\
    \subfloat[]{\includegraphics[width=0.37\textwidth]{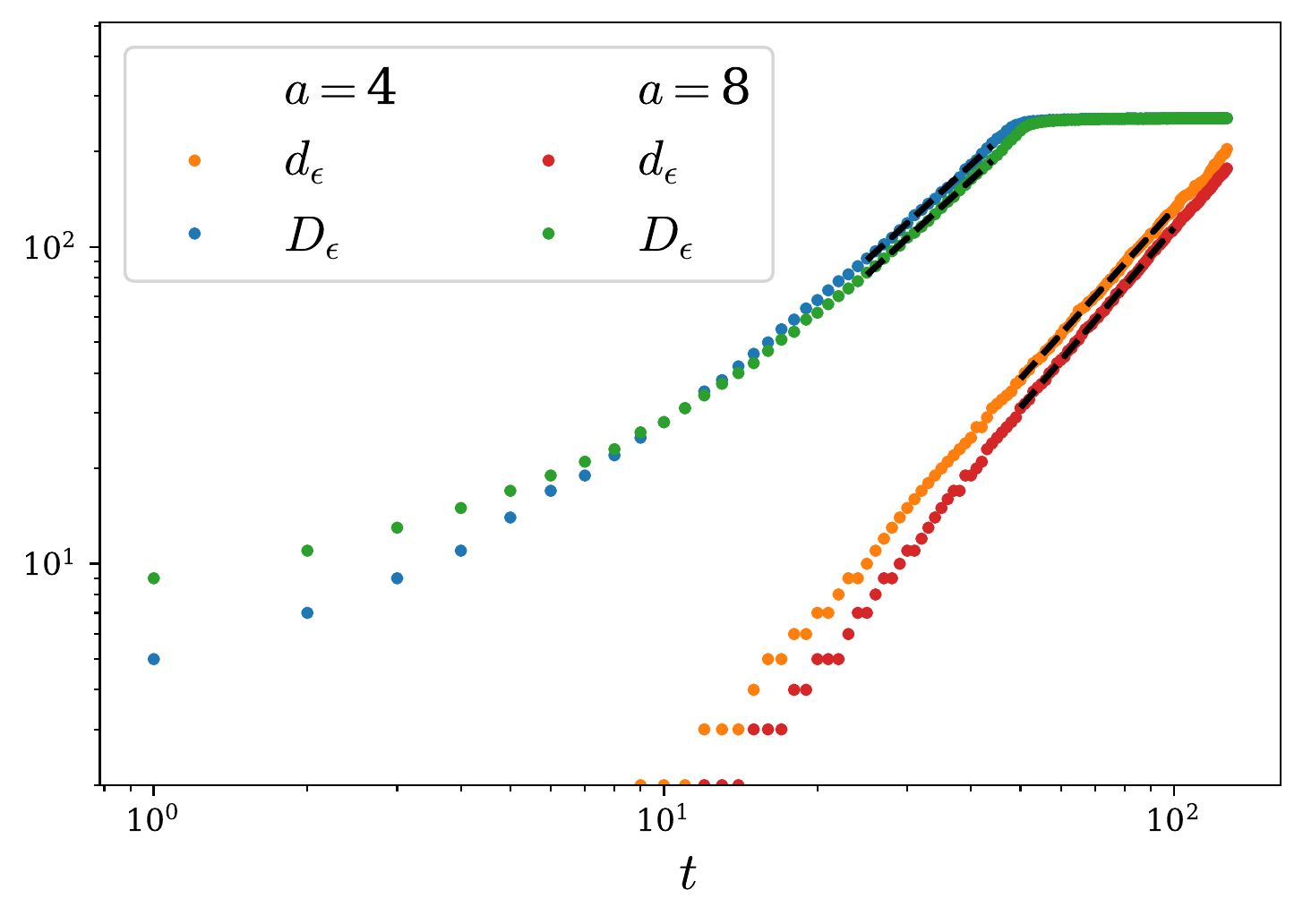}}
     \caption{Simulation results for $I_{A_R:B_P}(t)$ in a system with $L=768$ qubits. (a)(b) Heatmaps of the normalized mutual information $\mu_{A_R, B_P}$, as a function of $t$ and $|B|=2d$, for unitary and monitored dynamics respectively. (c) The growth of $d_\epsilon(t)$ and $D_\epsilon(t)$ with $\epsilon=0.01$. Both quantities develop power-law growing behaviors at late time. The $D_\epsilon$ curve has a plateau at late time as they reach the boundary of the simulated system.}
    \label{fig:spreading}
\end{figure}

Next, we study how the remaining $\mathcal{I}_A$ spreads with time. 
Since we are interested in the portion of $\mathcal{I}_A$ which was not lost by measurements, we analyze the following normalized mutual information:
\begin{equation}
    \mu(d,t) = \frac{I_{A_R:[-d,d]_P}(t)}{ I_{A_R: P}(t)}.
\end{equation}
Here, $\mu(d,t)$ measures the \textit{fraction} of survived $\mathcal{I}_A$ detectable within the region $[-d,d]$, at a given time $t$. Note that the spreading of $\mathcal{I}_A$ is only well defined before it is completely lost, \textit{i.e.} before the time scale \refeq{eq:timeloss}.

\reffig{fig:spreading} (a) and (b) show the heatmaps of normalized mutual information as a function of radius $d$ and time $t$, for unitary and monitored dynamics respectively. We find that there are three regimes: a dark blue region where $\mu(d, t)\approx 0$ for small $d$, a bright yellow region where $\mu(d, t)\approx 1$, and an intermediary transitioning region.

\begin{itemize}
\item $\mu(t, d)\approx 1$ implies that $[-d,d]$ contains all the remaining $\mathcal{I}_A$ at time t, and thus the smallest such $d$ marks the boundary of the region where the remaining $\mathcal{I}_A$ is supported on. Formally we can define this boundary as 
\begin{equation}
    D_\epsilon(t) = \mathrm{argmin}_d
    \left(\mu(d,t) \geq 1-\epsilon \right)
\end{equation}
for some $\epsilon\ll 1$, and $[-D_\epsilon(t), D_\epsilon(t)]$ is where the remaining $\mathcal{I}_A$ is suppoerted. We plotted $D_{0.1}(t)$ with orange dashed lines in \reffig{fig:spreading}(a) and (b). 

\item $\mu(t, d)\approx 0$ implies that one cannot probe any $\mathcal{I}_A$ within $[-d,d]$. We observe that, in both unitary and monitored cases, there exists a region with $\mu(t, d)\approx 0$. $[-d_\epsilon(t), d_\epsilon(t)]$ defines the largest region that no information $\mathcal{I}_A$ can be detected.
% This reflects a phenomenon called information spreading \cite{hayden2007black}: though all the remaining $\mathcal{I}_A$ is suppoerted on $[-D_\epsilon(t),D_\epsilon(t)]$, one can not recover any of $\mathcal{I}_A$ by only looking at a sub-extensive patch of the region
% \footnote{An alternative way to interpret \reffig{fig:spreading}(b) is that $\mathcal{I}_A$ supports only on the transitioning region $T(t)=[-D_\epsilon(t),-d_\epsilon(t)]\cup [d_\epsilon(t), D_\epsilon(t)]$ other than the whole $[-D_\epsilon(t), D_\epsilon(t)]$. But this is actually not the case, and can be ruled out by numerically verifying that $I_{A_R: T(t)_P}(t)\approx 0$ for any $t$.}. 
The radius of the information-less region can be quantified by ($\epsilon \ll 1$):
\begin{equation}
    d_\epsilon(t) = \mathrm{argmax}_d \left(\mu(d,t) \leq \epsilon \right).
\end{equation}
In \reffig{fig:spreading}(b), $d_{0.1}(t)$  are plotted using blue dashed lines. 
\end{itemize}

The physical meaning of $D_\epsilon(t)$ and $d_\epsilon(t)$ becomes clear by recalling the discussion in \refsec{sec:Troperators}. 
Let us assume that $O_A$ is transferred to $\tilde{O}_B$ by the dynamics $C_m$, in the sense of \refeq{eq:recoverlogical}. 
Recalling that $I_{A_R:B_P}$ measures the amount of logical operators transferred from $A$ to $B$, we conclude $\tilde{O}_B$ can be supported within $[-D_\epsilon(t),D_\epsilon(t)]$ but beyond $[-d_\epsilon(t), d_\epsilon(t)]$. \add{This is consistent with and further generalizes our conclusion in \refsec{sec:reversal}, stating that a $O(1)$ sized operator will be transferred to a $O(t^{3/2})$ sized one after $t$ steps of dynamics.}

The evolutions of the three regions define information lightcones, whose behaviors sharply differ between unitary and monitored dynamics. Here, $d_\epsilon(t)$ and $D_\epsilon(t)$ correspond to the inner and the outer front of the lightcone respectively. 
In the unitary case in \reffig{fig:spreading}(a), both lightcones grow linearly with time, consistent with emergent causality. 
In contrast, in the monitored case in \reffig{fig:spreading}(b), both lightcones are outward bent, corresponding to the super-linear spreading of $\mathcal{I}_A$. 
In \reffig{fig:spreading}(c), we plot the growth of $d_\epsilon(t)$ and $D_\epsilon(t)$ with $\epsilon=0.01$. 
Both $d_\epsilon(t)$ and $D_\epsilon(t)$ display power-law growth in late time, with the fitted exponents (fitted lines are displayed as black dashed lines in the figure):
\begin{equation}\label{eq:exponents}
    \begin{aligned}
        a=4: 
        \quad
        d_\epsilon(t)\propto t^{1.78},
        \quad
        D_{\epsilon}(t)\propto t^{1.47},\\
        a=8:
        \quad
        d_\epsilon(t)\propto t^{1.89},
        \quad
        D_{\epsilon}(t)\propto t^{1.46}.
    \end{aligned}
\end{equation}

We notice that the inner and the outer front have different exponents. 
Namely, the latter one is close to $3/2$ as predicted in \refsec{sec:reversal}. 

Furthermore, the inner exponent is larger than the outer one, suggesting that two fronts are likely to collide at some point. 
This is consistent with the interpretation that the inner front and the outer front respectively mark the shortest and the longest survived logical operator initially supported on $A$. 
Recall that all the logical operators will be eventually destroyed (\textit{i.e.} $I_{A_R: P}\approx 0$) at a time scale $t\propto |A|^2$, and shortly before this time scale, there will be only one logical operator surviving. 
This logical operator is both the longest logical operator and the shortest logical operator, which means that two fronts of the lightcone collide. 
Numerically accessing this time scale, however, is challenging due to the requirement of running a large system for a long time. 

We currently do not have a physical explanation of the exponent for $d_\epsilon(t)$. We leave this as a future question.

\subsection{Partially fixing the input state}\label{sec:alt-way}

Until now, we have treated $C_{\textbf{m}}$ as a dynamical process whose input is the entire physical system $P$. 
In this subsection we discuss the consequence of using only a small subregion $A\subseteq P$ as the input while fixing the initial state for other qubits in $A^c=P\backslash A$. 
In the error-correcting code perspective, fixing part of the initial state effectively selects a subspace of the original code space.
Here, we restrict our discussion to the case that $\rho_{A^c}$ is a homogeneous product state: $\rho_{A^c}=\otimes_{i\in A^c} (\rho_0)_{i}$, where $\rho_0$ is some single qubit density matrix, see \reffig{fig:pure_initial_state}(a) for illustration. 
A similar setting was considered in~\cite{ippoliti2021entanglement} for studying information spreading in measurement-only circuits.

\begin{figure}[h]
% \captionsetup[subfigure]{aboveskip=-2pt,belowskip=-2pt}
\centering
\subfloat[]{\includegraphics[width=0.28\textwidth]{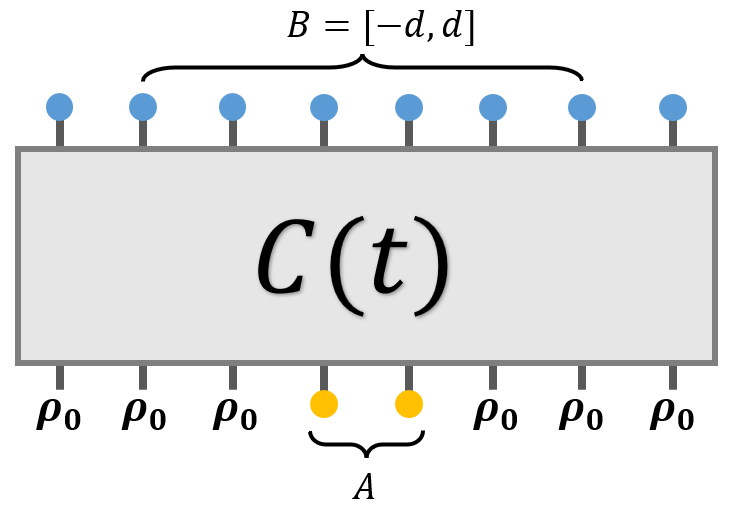}}\\
\subfloat[]{\includegraphics[width=0.35\textwidth]{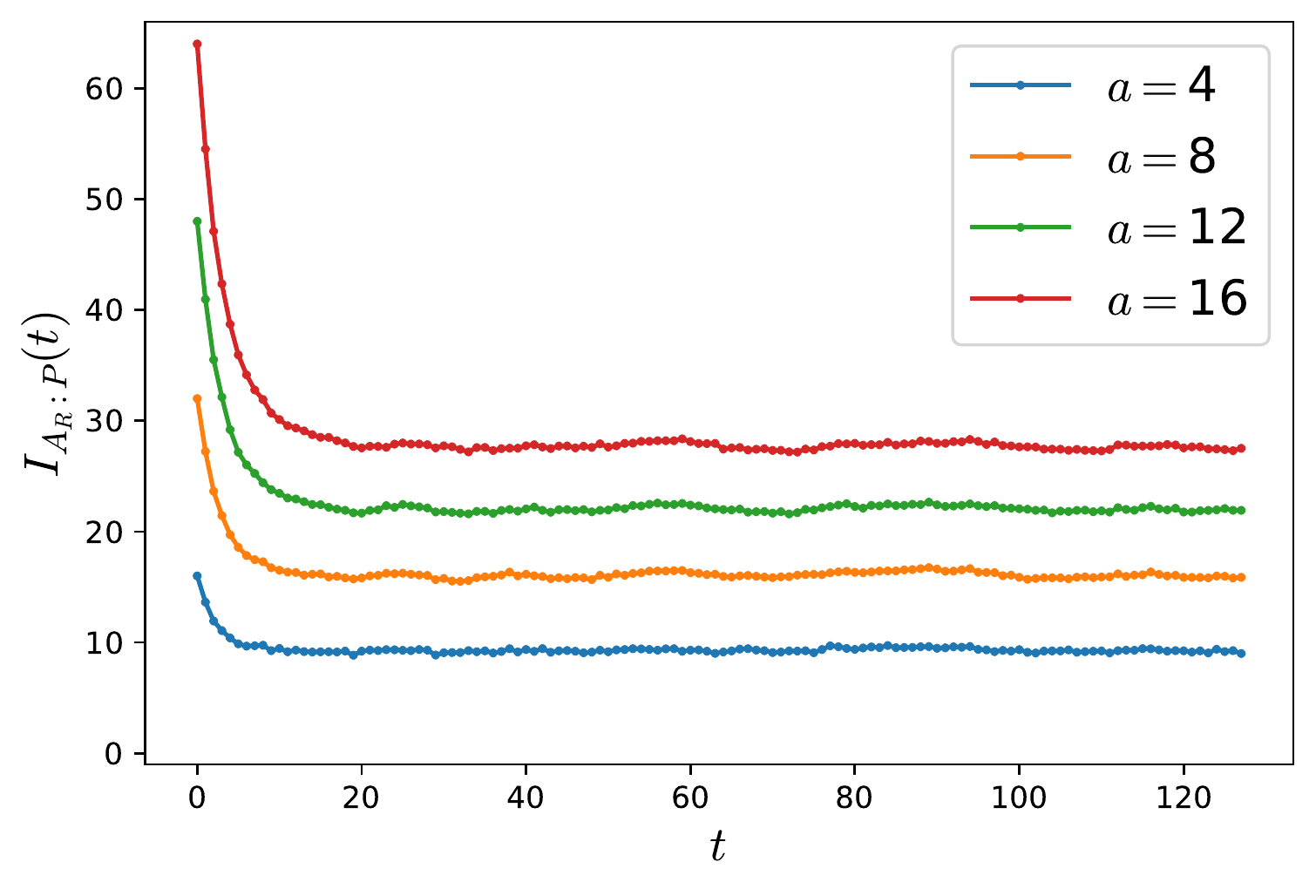}}\\
\subfloat[]{\includegraphics[width=0.40\textwidth]{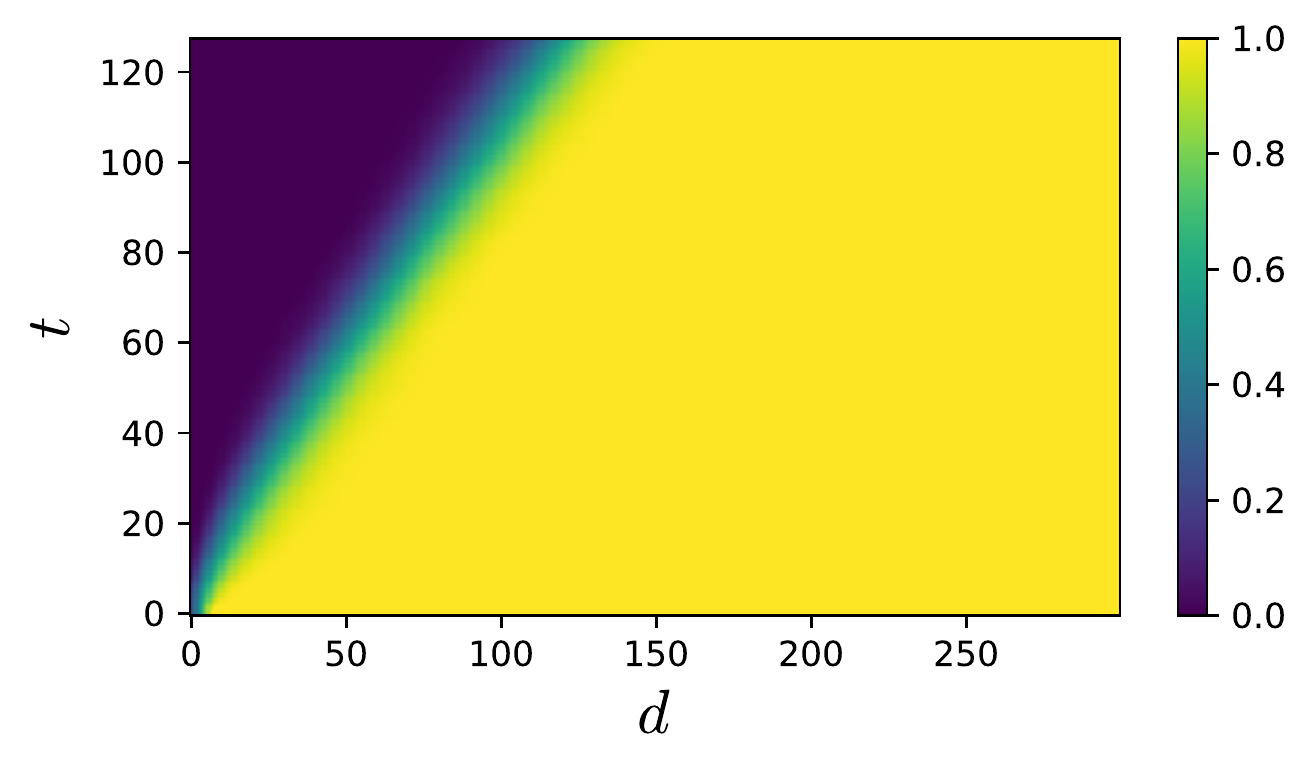}}
\caption{(a) The circuit configuration for the alternative setting to characterize information spreading discussed in \refsec{sec:alt-way}, where $A=\{3,4\}$ (b) Decay of $A=[-a,a]$'s initial information with time $t$. (c) Heatmap of the normalized mutual information $\mu(d,t)$ in the alternative setting}
\label{fig:pure_spreading}
\end{figure}

Let us begin by introducing a reference system $R$ which is maximally entangled with $A$ initially. 
Note that, in this case, $R=A_R=[-a,a]_R$ has the same size as $A$. 

We first discuss the case that $\rho_0=\frac{1}{2}\mathbb{I}$, \textit{i.e.} the input state on $A^c_P$ is initially maximally mixed. 
This enforces no constraint on the codespace, so it is equivalent to the previously studied case in \refsec{sec:wavefront}.

Next, we consider the case $\rho_0=\ket{0}\bra{0}$, \textit{i.e.} qubits in $A^c_P$ are initially prepared in a pure product state. 
We present the numerically simulated $I_{A_R:B_P}(t)$ in \reffig{fig:pure_spreading}(b-c). 
In panel (b), we observe that $A$'s input information equilibrates to some constant value instead of decaying to zero with time (at exponentially long time, we expect it to decay to zero). 
In addition, in panel (c) which displays the heatmap for the normalized mutual information $\mu(d,t)$, we find that the remaining information spreads only linearly, instead of super-linearly with time. 
This behavior of information spreading is similar to the one observed in measurement-only dynamics within its volume-law, as reported in \cite{ippoliti2021entanglement}.

\section{domain-wall Picture}\label{sec:domain-wall}

It has been proposed that several universal features of entanglement dynamics in weakly monitored circuits can be described effectively by modeling entanglement domain-walls using directed polymers in a random environment (DPRE)~\cite{li2021dpre}.
In this section, we revisit our results from the domain-wall picture and the DPRE effective theory.

\subsection{Review of domain-wall picture}\label{sec:review}

We begin by providing a brief review of the domain-wall picture and the DPRE effective theory. \add{We refer interested readers to \cite{jian2020measurement, bao2020theory, li2021statistical} for detailed discussions on how to map monitored circuits to statistical mechanics models as well as how the former's entanglement properties are related to domain-walls in the latter.}

Consider a $t$-layer circuit defined on an infinite qubit chain with the initial state being maximally entangled with the reference $R$ as in \reffig{fig:two_copy}. 
By stacking $Q$ copies of the $(1+1)$-dimensional circuit and performing disorder average over random unitary gates, one can map the whole system to a $(2+0)$-dimensional classical statistical mechanics model on the same space-time manifold. 
Through this mapping, each unitary gate is mapped to a spin in the bulk, while each qubit in $P$ and $R$ is mapped to a spin located at the boundaries. 
Here, spins $\{s\}$ take values in the permutation group $S_Q$. 
The weakly-monitored phase of the circuit corresponds to the ordered phase of the spin model. 

Entanglement properties of the original quantum circuit can be found by looking for the domain-walls in the statistical mechanics model.
In particular, the entanglement entropy $S_A(t)$ of a boundary region $A\in P\cup R$ is given by the free energy cost of a domain-wall that separates $A$ from the rest of the boundary $\bar{A}=P\cup R-A$:
\begin{equation}\label{eq:entropy_and_DW}
    S_A(t) \approx f\left(\eqfig{2cm}{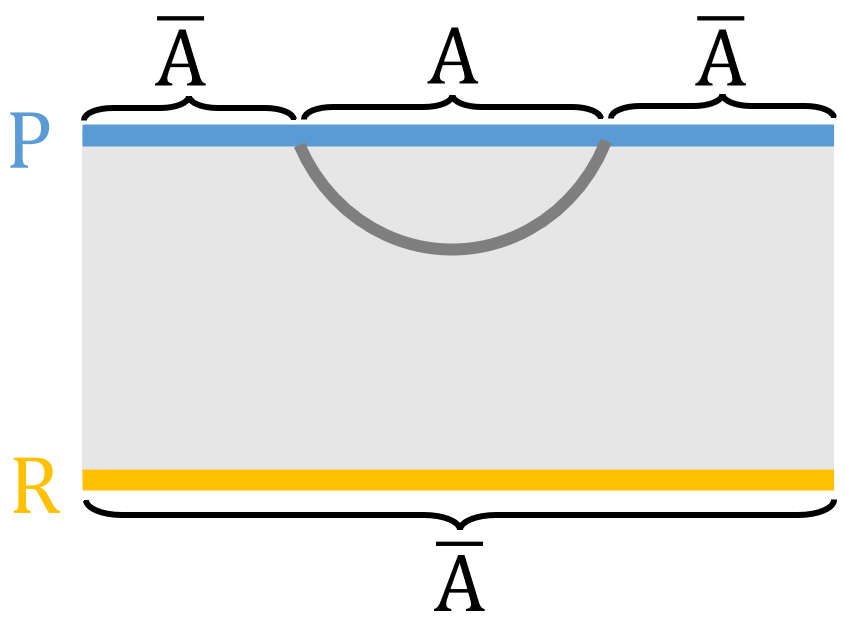}\right)
\end{equation}
where $f(\cdot)$ denotes the energy cost of the domain-wall(s) configuration (the dark grey curve). 

While the above figure considers the case where $A$ is a single contiguous region, the region $A$ can be any subset of $P\cup R$. 
\add{Generally}, there may be domain-wall configurations with different topologies that separate $A$ and $\bar{A}$. For example, if $A=A_1\cup A_2$ with $A_1\subseteq P$ and $A_2\subseteq R$, the domain-wall can have two different topologies:
\begin{equation*}
    \eqfig{1.8cm}{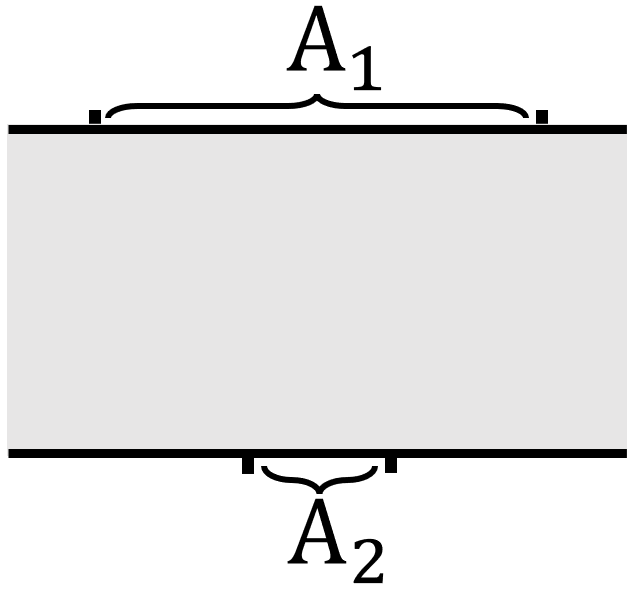}:\quad\eqfig{1.2cm}{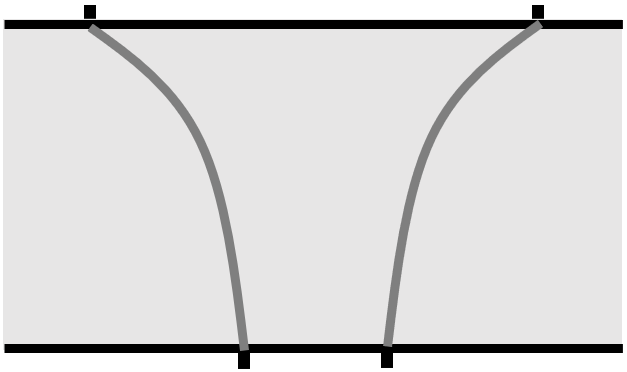}\quad \eqfig{1.2cm}{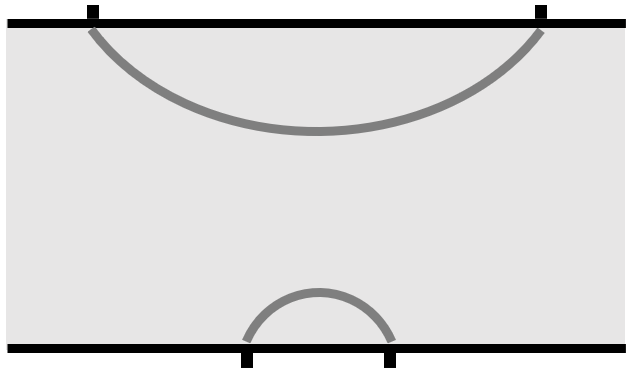}
\end{equation*}
In this case the free energy, in the leading order, is given by the configuration that gives the smaller free energy:
\begin{equation}
    S_A(t)=\min\left\{f\left(\eqfig{0.7cm}{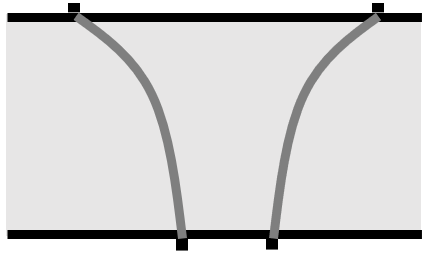}\right),  f\left(\eqfig{0.7cm}{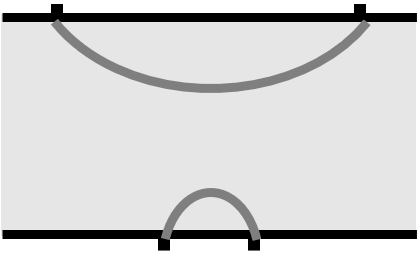}\right)\right\}
\end{equation}
\add{Furthermore, if the domain-wall has several disconnected pieces (for instance in the case shown above), its total free energy is the sum of each individual piece's contribution.}

% Furthermore,
\add{Solving domain-wall's properties from the statistical mechanics model is in general not feasible. In \cite{li2021dpre}, the authors find numerical and analytical evidence suggesting that large-scale properties of the domain-wall are captured by DPRE. In particular,}
a single piece of domain-wall's energy cost is approximately proportional to the length of the shortest path connecting two endpoints in a random media:
\begin{equation}\label{eq:randpolymer}
    \begin{aligned}
        f\left(\eqfig{.75cm}{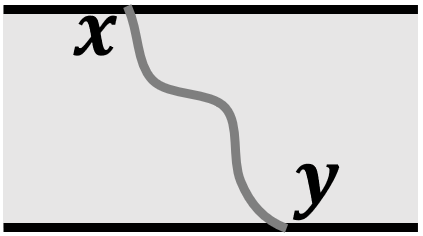}\right) &\simeq \min_{\substack{\textbf{z}(\tau): \textbf{z}(0)=\textbf{x},\\ 
        \textbf{z}(1)=
        \textbf{y}}}\left(\int_0^1 V(\textbf{z}(\tau)) |\textbf{z}'(\tau)|\d \tau \right)\\
        &\equiv D(\textbf{x}, \textbf{y}),
    \end{aligned}
\end{equation}
where $V(\textbf{x})$ is a quenched random potential without spatial-temporal correlation $\text{cov}(V(\textbf{x}),V(\textbf{y}))=\sigma^2\delta_{\textbf{x},\textbf{y}}$ and with a positive mean $\E[V(\textbf{x})]=\mu>0$. The $\textbf{z}(\cdot)$ runs over all paths connecting $\textbf{x}$ to $\textbf{y}$. It should be noted that, while the leading term in $D(\textbf{x}_1, \textbf{x}_2)$ is proportional to $|\textbf{x}_1-\textbf{x}_2|$, it contains a subleading contribution which gives rise to non-trivial multipartite entanglement as we shall see later. 

\add{
Before finishing the review, we discuss an issue about the directedness of the shortest paths $\textbf{z}(\tau)$ appearing in \refeq{eq:randpolymer}. In the original proposal of the effective theory~\cite{li2021dpre}, it is further required that $\textbf{z}(\cdot)$ is only chosen from paths that are \textit{directed} (\textit{i.e.} not going backward) along the $\textbf{x}$-$\textbf{y}$ direction, as is suggested by the word `directed' in the name DPRE. In fact, this extra constraint does not have to be added by hand. This is because for a typical realization of $V$, the not-necessarily-directed shortest path connecting $\textbf{x}$ and $\textbf{y}$ has only a $O(1)$ length overhang along the $\textbf{x}$-$\textbf{y}$ direction \cite{huse1985pinning}, which is negligible after coarse-graining. In other words, directed shortest paths and un-directed shortest paths connecting $\textbf{x}$ and $\textbf{y}$ belong to the same universality class, and are indistinguishable in the thermodynamic limit. Thus, depending on our purpose, we can impose or lift the directedness at our convenience at the lattice level without worrying about getting different conclusions in the thermodynamic limit.
}

\subsection{Domain-wall interpretation of $S_A$}

Henceforth, for each quantity $O$ defined in the context of monitored circuits, we use $[O]$ to denote its counterpart in the DPRE effective theory.

We start from the simplest case: entanglement entropy of a single contiguous interval $A$. Using the domain-wall description \refeq{eq:entropy_and_DW}, $S_A$ can be written as follows:
\begin{equation}
        \eqfig{1.3cm}{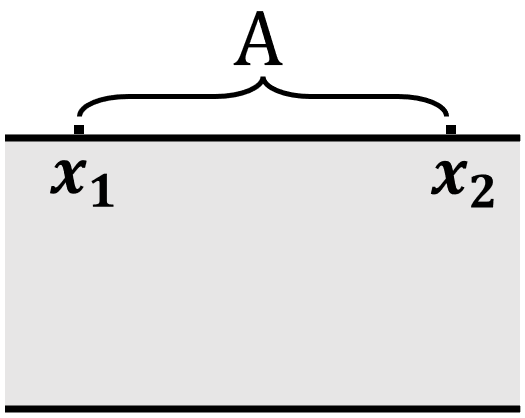}:\ \ S_A(t) = f\left(\eqfig{.8cm}{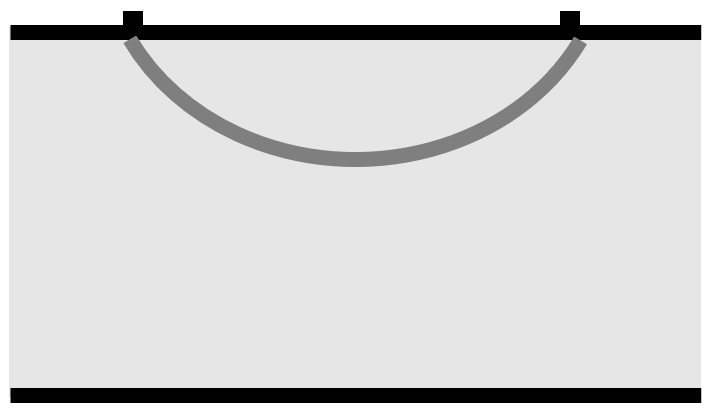}\right),
\end{equation}
where
\begin{equation}\label{eq:cor1}
\textbf{x}_1 = (0, t),\quad \textbf{x}_2 = (|A|, t)
\end{equation}
and 
\begin{equation}\label{eq:S_poly}
    [S_A(t)] \simeq D(\textbf{x}_1, \textbf{x}_2).
\end{equation}
Numerical and theoretical studies in \cite{li2021dpre} found that the mean of $D$ takes the following scaling form:
\begin{equation}
    \overline{D(\textbf{x}_1, \textbf{x}_2)} = s_0(t)|A| + |A|^{\frac{1}{3}}\Phi(t|A|^{-\frac{2}{3}}),
\end{equation}
where $\Phi(\cdot)$ in the subleading term is a universal scaling function with the following limiting behavior:
\begin{equation}
    \Phi(\eta)=
    \left\{
    \begin{array}{ll}
        \eta^{1/2} &\eta\rightarrow 0  \\
         \eta^{0} &\eta\rightarrow \infty
    \end{array}
    \right.
\end{equation}
This implies that the subleading part of $S_A(t)$ behaves as:
\begin{equation}
    S^{\mathsf{sub}}_A(t) \simeq
    \left\{
    \begin{array}{ll}
        t^{\frac{1}{2}} &t \ll |A|^{\frac{2}{3}} \\
        |A|^{\frac{1}{3}} &t \gg |A|^{\frac{2}{3}}
    \end{array}
    \right.
\end{equation}
This agrees with our results derived from the stabilizer length distribution in Sec. \ref{sec:S_cont}. Furthermore, since the conditional mutual information \refeq{eq:cmiscaling} and the growth of contiguous code distance \refeq{eq:dcode1} rely only on the dynamics of $S_A(t)$ in single intervals, we conclude these are also consistent with the DPRE effective theory.

The domain-wall picture also provides a useful insight into the $S_A$'s different behaviors for pure and for completely mixed initial states. In the pure initial state case, the underlying statistical mechanics model assumes free boundary condition on the $t = 0$ boundary \cite{jian2020measurement, li2021conformal, bao2020theory}, and domain-walls can freely start or end on the $t=0$ boundary. We thus have
\begin{equation}
    S_{A}(t) \approx \min\left\{f\left(\eqfig{0.7cm}{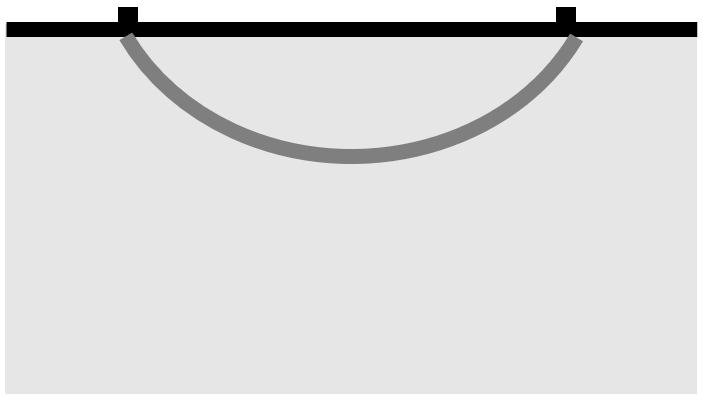}\right), f\left(\eqfig{0.7cm}{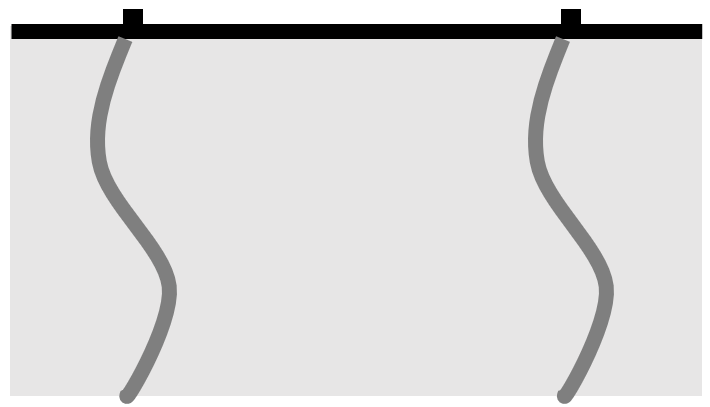}\right)\right\}.
\end{equation}
Using properties of the $D(\cdot)$ function, we notice that in the leading order the first term is $O(|A|)$, while the second term is $O(t)$. This suggests that, when $t$ is sublinear in $|A|$, the second term is preferred over the first one. 
This changes at a time scale $t^* \simeq |A|$, after which the first term becomes smaller. Furthermore, we notice that the time scale $t^*$ is far later than the other time scale $\simeq |A|^{2/3}$ at which the first term would reach the equilibrium. We thus conclude that $S_A(t)$ hardly changes after $t^*$. In summary, we have the following behavior of $S_A(t)$:
\begin{equation}
        S_A(t) \simeq 
    \left\{
    \begin{array}{lll}
         t & &t<t^*=O(|A|) \\
         |A| & &t>t^*
    \end{array}
    \right.
\end{equation}
which agrees with the results we obtained in \refsec{sec:pure_vs_mix}.

\subsection{Domain-wall interpretation of $I_{A:B}$}\label{sec:dwIAB}

Next, we turn our attention to the mutual informaiton $I_{A:B}$ and present a derivation of \refeq{eq:mutualinfoapp_0} using the domain-wall picture, which is reprinted below
\begin{equation}
I_{A:B}\approx\max\{S_{A}+S_{B}-S_{ABC}-S_C, 0\}.
\end{equation}
Here, $A$, $B$ and $C$ are arranged as follows:
\begin{equation*}
    \eqfig{2cm}{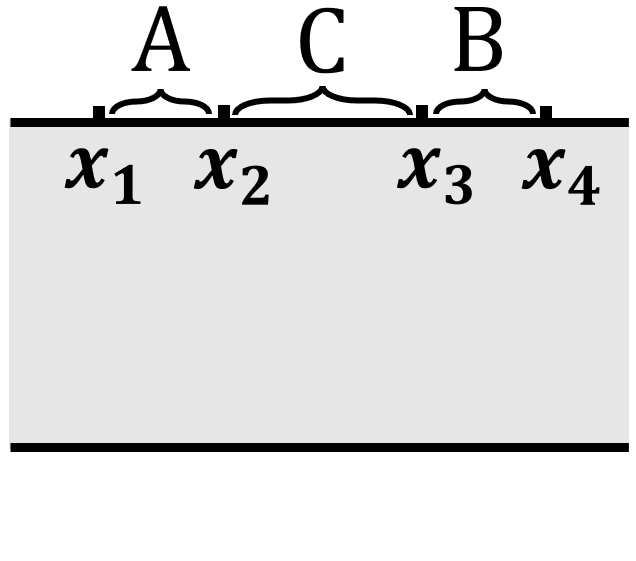}
\end{equation*}

Using the relation \refeq{eq:entropy_and_DW}, the domain-wall representations for the three terms in $I_{A,B}$ are:
\begin{equation}\label{eq:dmIAB1}
    \begin{aligned}
        S_A&= f\left(\eqfig{.8cm}{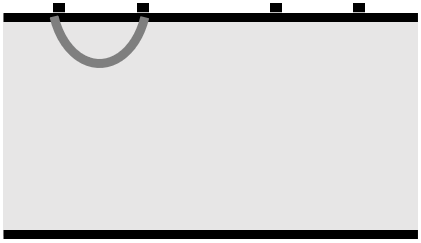}\right),\\
        S_B&= f\left(\eqfig{.8cm}{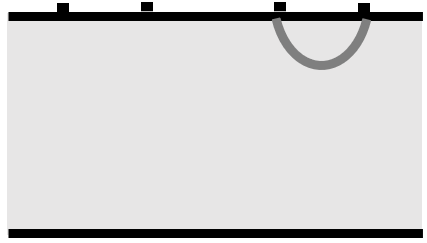}\right),\\
        S_{AB}&= \min\left\{f\left(\eqfig{.8cm}{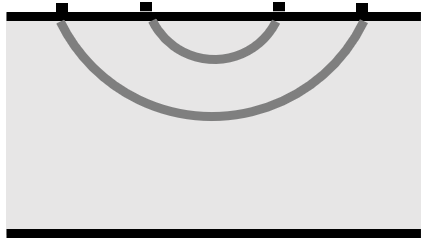}\right), f\left(\eqfig{.8cm}{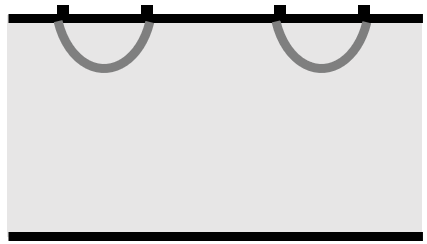}\right)\right\}\\
        &=\min\{S_{ABC}+S_C, S_A+S_B\}.
    \end{aligned}
\end{equation}
Combining them together, we again obtain \refeq{eq:mutualinfoapp_0}.

Here we would like to point out a certain subtlety, due to the subleading contributions, in evaluating the minimum of domain-wall's free energies. Recalling that the leading order contribution of a domain-wall's free energy is proportional to the distance between two endpoints, one might think that the second domain-wall configuration for $S_{AB}$ in \refeq{eq:dmIAB1} would be always smaller. Notice, however, that the difference in the leading order contributions is $O(|C|)$, and if $|C|$ is comparable to the subleading terms in $S_A$ or $S_B$, the subleading contributions could make the first domain-wall configuration smaller. This is indeed what we observed in the calculations of the code distance in \refsec{sec:coded}.

\subsection{Domain-wall interpretation of $I_{A_R:B_P}$}

Finally, we evaluate $I_{A_R:B_P}$, defined in \refeq{eq:mi}, which was used as a measure of information spreading in \refsec{sec:info-spread}. Recall that the intervals $A$ and $B$ have the following configuration:
\begin{equation*}
        \eqfig{1.6cm}{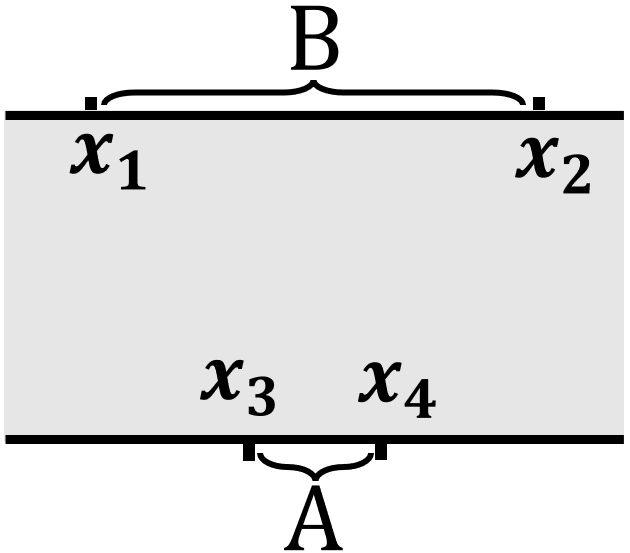},
\end{equation*}
Using the domain-wall picture, we can obtain the following relations:
\begin{equation}\label{eq:dw-mi}
\begin{aligned}
    S_{A_R} &= f\left(\eqfig{.6cm}{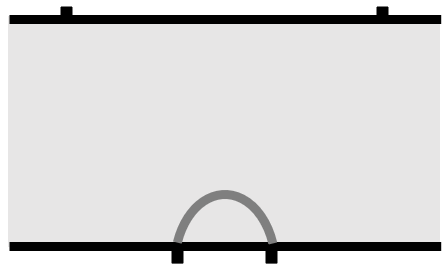}\right)\\
    S_{B_P} &= f\left(\eqfig{.6cm}{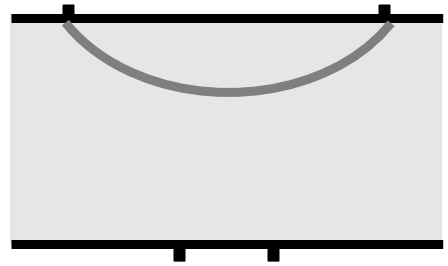}\right)\\
    S_{A_R\cup B_P} &= \min\left\{f\left(\eqfig{.6cm}{figs/fig_IAB2}\right), f\left(\eqfig{.6cm}{figs/fig_IAB1}\right)\right\}\\
    \Rightarrow I_{A_R:B_P} &\approx \max\left\{0, f\left(\eqfig{.6cm}{figs/fig_IAB2}\right)-f\left(\eqfig{.6cm}{figs/fig_IAB1}\right)\right\}.
\end{aligned}
\end{equation}
In the last line we make use of the additive property of the domain-walls' energies:
\begin{equation}
    f\left(\eqfig{.6cm}{figs/fig_S_down}\right)+f\left(\eqfig{.6cm}{figs/fig_S_up}\right)\approx f\left(\eqfig{.6cm}{figs/fig_IAB2}\right) .
\end{equation}
The four points involved are chosen as
\begin{equation}
    \begin{aligned}
        &\textbf{x}_1=(-d/2, t),
        &\textbf{x}_2=(d/2, t),\\
        &\textbf{x}_3=(-a/2, 0),
        &\textbf{x}_4=(a/2, 0).
    \end{aligned}
\end{equation}
and we have
\begin{equation}\label{eq:IAB_DW}
\begin{aligned}
    [I_{A_R:B_P}(t)] = \max\{&0, D(\textbf{x}_1, \textbf{x}_2) + D(\textbf{x}_3, \textbf{x}_4) \\
    &- D(\textbf{x}_1, \textbf{x}_3) - D(\textbf{x}_2, \textbf{x}_4)\}.
\end{aligned}
\end{equation}

Since the analytical expressions for $D(\textbf{x}_1, \textbf{x}_3)$ and $D(\textbf{x}_2, \textbf{x}_4)$ are not known, we resort to a direct simulation of \refeq{eq:IAB_DW}. Details about the shortest paths simulation can be found in \refsec{ap:dpre-sim}. 

As shown in \reffig{fig:randpolymer}(a), the behavior of the simulated $[\mu_{A_R: B_P}(t)]=[I_{A_R: B_P}(t)]/[I_{A_R:
P}(t)]$ is qualitatively similar to that of $\mu_{A_R: B_P}$ shown in \reffig{fig:spreading}(b). Furthermore, we calculate $[d_\epsilon(t)]$ and $[D_\epsilon(t)]$ and find that they also develop power-law growth in late times (\reffig{fig:randpolymer}(b), black dashed lines are power-law fittings). The exponents are fitted to be:
\begin{equation}\label{eq:exponents2}
        d_\epsilon(t) \propto t^{1.92},
        \quad
        D_{\epsilon}(t) \propto t^{1.53}.
\end{equation}
The exponents numerically agree with those obtained in \refeq{eq:exponents}, reflecting the consistency between our results and the DPRE effective theory.
\begin{figure}[h]
    \centering
    \subfloat[]{\includegraphics[width=0.40\textwidth]{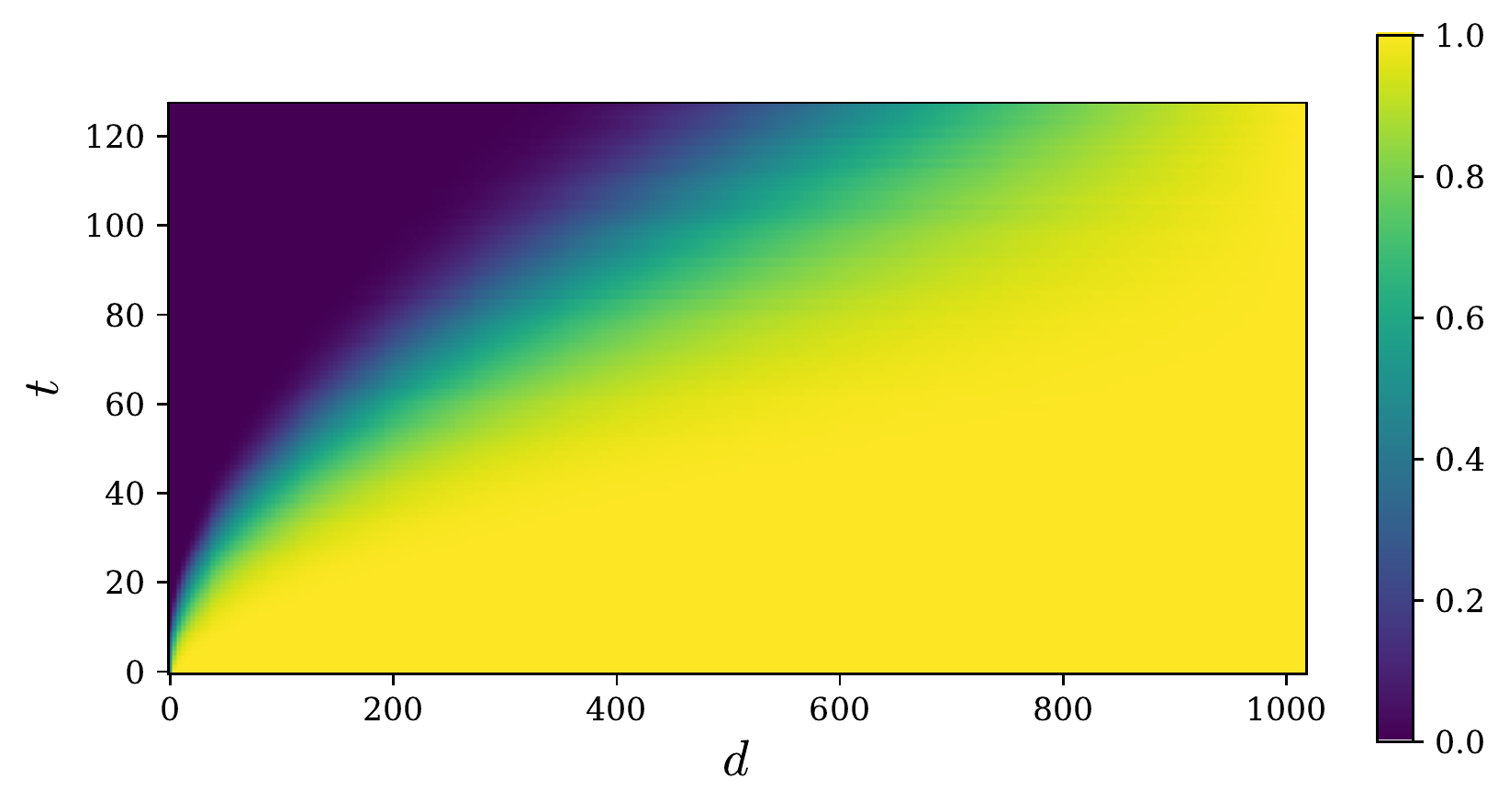}}\\
    \subfloat[]{\includegraphics[width=0.38\textwidth]{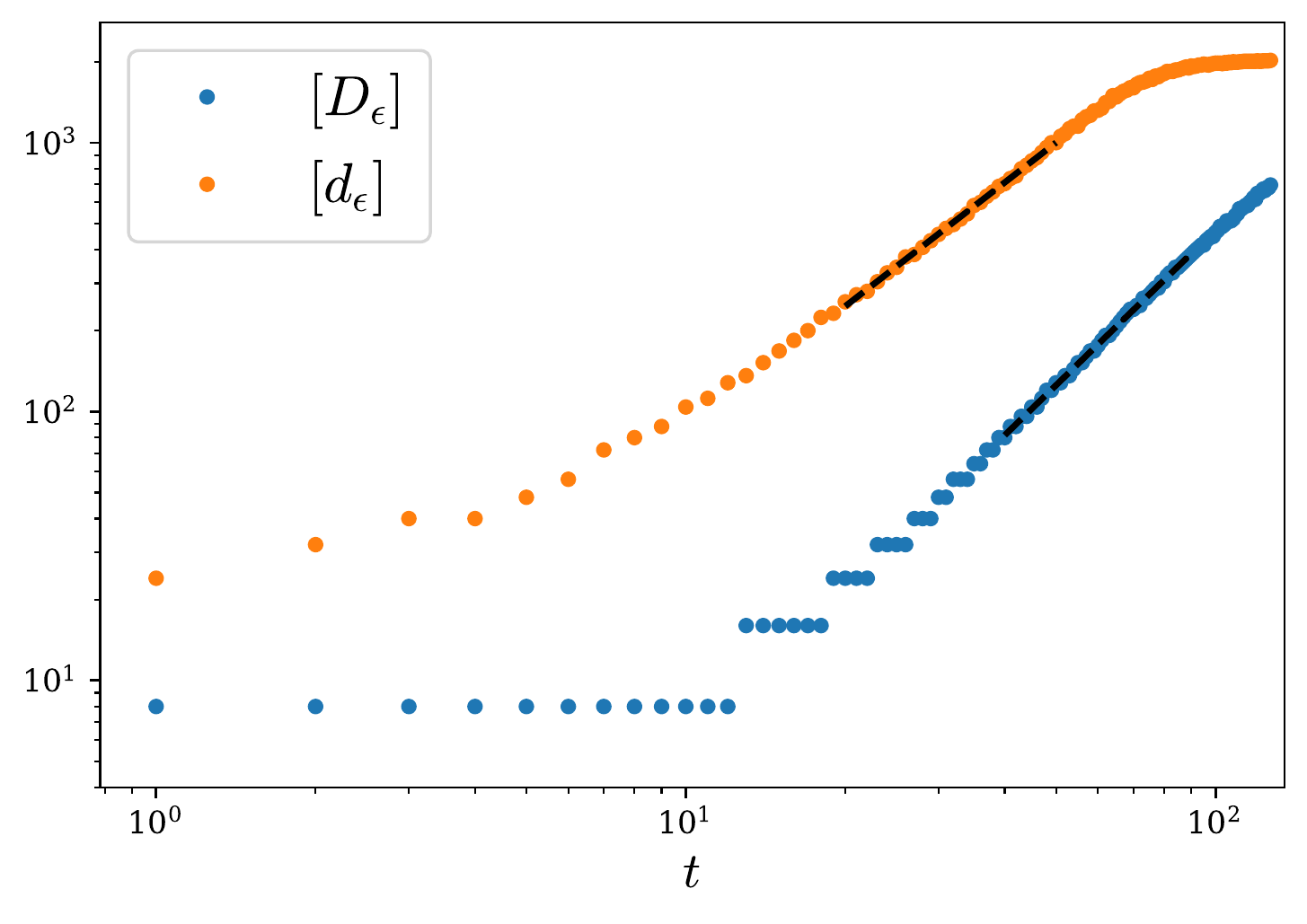}}
    \caption{Numerical result for the shortest path simulation. (a) Heatmap of the simulated $[\mu_{A_R, B_P}(t)]$. The figure should be compared with its counterpart in the Clifford simulation \reffig{fig:spreading}(b). (b) $[d_\epsilon(t)]$ and $[D_\epsilon(t)]$ extracted from the simulation data, with $\epsilon=0.01$. The black dashed lines are power-law fittings, with fitted exponents shown in \refeq{eq:exponents2}}
    \label{fig:randpolymer}
\end{figure}

\section{Discussion}\label{sec:discussion}

In this work, we studied the quantum information dynamics in weakly monitored Clifford quantum circuits and found that entanglement and information can spread superlinearly with time ($\simeq t^{\frac{3}{2}}$). 
Namely, we observed that such superballistic propagation of information is mediated by superlinear growth of the size of stabilizer generators due to projective measurements.  
Furthermore, this led to a new sublinear time scale of $\simeq L^{\frac{2}{3}}$ which can be interpreted as the encoding time of a dynamical quantum error-correcting code. 
It is important to emphasize again that these nonlocal effects emerge in the non-relativistic limit where each local observer instantly learns the measurement outcomes. 

While we focused on a particular setup of a $(1+1)$-dimensional random weakly-monitored Clifford circuit with $p=0.08$, we expect that our results reveal universal aspects of monitored quantum many-body dynamics.
Indeed, some of our technical results suggest further generalization. For example, the stabilizer length distributions can be replaced with the conditional mutual information \refeq{eq:clipCMI} when studying monitored circuits beyond Clifford gates. 

One important future question is to reveal the microscopic origin of the superlinear entanglement growth. In this work, we have taken a few first steps toward this goal. Namely, we identified the scaling form of the time evolution of the stabilizer length distributions $h(\ell,t)$ for both pure and mixed initial states. A key next step will be to understand the underlying microscopic mechanism for the superlinear growth of stabilizer generators and the effect of projective measurements. Also, we have presented a derivation of the mutual information for two disjoint intervals by utilizing a probabilistic argument on the operator contents of stabilizer generators. Our analysis matched with the macroscopic prediction from DPREs, including the subleading contributions. This hints that probabilistic arguments on operator contents may provide further insights into the microscopic origin of the effective theory of monitored quantum circuits. 

The notion of information spreading has played a key role in investigating the bulk and boundary dynamics in the AdS/CFT correspondence. This naturally prompts us to ponder over a possible gravitational dual of monitored quantum circuits and the geometric manifestation of the superlinear entanglement growth. In the bulk quantum gravity, it has been conjectured that some class of projection operations can be modeled as insertions of End-of-World (EoW) branes~\cite{almheiri2018holographic, takayanagi2011holographic}. In the boundary quantum many-body physics, recent studies suggest that projective measurements of a few qubits can induce significant changes of the entanglement structure due to the underlying scrambling dynamics~\cite{yoshida2022projective}. The effect of measurements on a conformal field theory has been considered in \cite{Rajabpour_2016, MIE, altman_critical}, and the holographic bulk interpretation analyzed in \cite{swingle1, swingle2}. It will be interesting to further look into the geometric interpretation of continuous monitoring in quantum gravity and the resulting entanglement dynamics.

\begin{acknowledgements}
We thank Yaodong Li for helpful discussions, as well as Roger Melko for computational resources. This work
was supported by Perimeter Institute, NSERC, and Compute Canada. Research at Perimeter Institute is supported in part by the Government of Canada through the Department of Innovation, Science and Economic Development and by the Province of Ontario through the
Ministry of Colleges and Universities.

\end{acknowledgements}

\bibliography{main.bib}

\appendix

\section{Monitored dynamics as a quantum channel}\label{sec:mixchannel}

In the main part of the paper, we focused on the trajectory dynamics $C_{\textbf{m}}$ corresponding to a single measurement outcome sequence $\textbf{m}$. In this appendix, we study the ensemble of all the trajectory dynamics $\{C_{\textbf{m}}\}$ that can arise in a monitored circuit with given unitary gates and measurement locations.

To study the ensemble of dynamics $\{C_{\textbf{m}}\}$, we introduce a set of registers $M$ to record all the measurement outcomes. The whole process can be viewed as a Clifford quantum channel from $P$ to $P\cup M$ ~\cite{gullans2020dynamical, choiQEC}, written as:
\begin{equation}\label{eq:monitored_channel}
    \channel_C(\rho_P) = \sum_{\textbf{m}\in\{0,1\}^{|M|}} \ket{\textbf{m}}\bra{\textbf{m}}_M \otimes (C_{\textbf{m}}\rho C_{\textbf{m}}^\dagger)_P.
\end{equation}
Here we have made each measurement outcome decoherent by applying a complete dephasing channel on registers.

Given a Pauli operator $O_{A}$ defined on $A\subset P$, we are naturally led to look for a Pauli operator $\tilde{O}_S$ defined on $S\subset Q$ such that:
\begin{equation}\label{eq:recovery}
    \tr(O_A\rho) = \tr(\tilde{O}_S\channel_C(\rho)) \quad \forall \rho.
\end{equation}
Accordingly, we use the following quantity as a measure for the amount of information transferred from $A$ into $S$ by the stabilizer channel $\channel$:
\begin{equation}\label{eq:relation1}
    H_{A\rightarrow S}(\channel)
    \equiv 
    |\{O\in \pauli_A:\exists~\tilde{O}_S\in\pauli_S\ s.t.\ O = \channel^\dagger(\tilde{O}_S)\}|,
\end{equation}
where $\pauli_{A}$ ($\pauli_{S}$) denotes the set of Pauli operators supported on $A$ ($S$). 

The central result of this appendix is a relation between the entanglement properties of a single trajectory dynamics and that of the ensemble of dynamics.

\begin{theorem} \label{thm:channelinfo}
Letting $\channel_C$ be a Clifford monitored dynamics and $A,B\subseteq P$, we have: 
\begin{equation}\label{eq:monitored_relation}
         I_{A_R:B_P}(\ket{\phi_{\textbf{m}}})
         = \log_2 H_{A\rightarrow B\cup M}(\channel_C) - \log_2 H_{A\rightarrow M}(\channel_C).
\end{equation}
\end{theorem}

On the \textit{r.h.s} of the equation, the first (second) term measures how much information is transmitted from $A$ to $B\cup M$ (to $B$). Therefore,  $I_{A_R:B_P}(t)$ quantifies  the amount of \textit{extra} information within the input $A$ extractable within the output $B$, given all the classical measurement outcomes $M$ are known. Below, we present a proof of Thm.\ref{thm:channelinfo}. 

It is convenient to introduce the Choi state:
\begin{equation}\label{eq:ChoiPhi}
    \begin{aligned}
        \Phi_\channel
        =&\frac{1}{\dim P}
        \sum_{i, j}\channel_C(\ket{i}\bra{j}_P)\otimes \ket{i}\bra{j}_R\\
        =&
        \sum_{\mathbf{m}\in\{0,1\}^{|M|}}p_{\mathbf{m}}(\ket{\mathbf{m}}\bra{\mathbf{m}})_M\otimes (\ket{\phi_{\mathbf{m}}}\bra{\phi_{\mathbf{m}}})_{PR}.
    \end{aligned}
\end{equation}

Stabilizer generators of the Choi state $\Phi_\channel$ can be related to the recoverability of operators via the following lemma.

\begin{lemma}\label{lemma:Phistab}
Given a stabilizer channel $\channel:\opspace(\hilbert_P)\rightarrow\opspace(\hilbert_Q)$, $O\in\pauli_P$ is recovered by $\tilde{O}\in\pauli_Q$ through $\channel$ if and only if $O^T\otimes \tilde{O}$ is a stabilizer of $\Phi_\channel$.
\end{lemma}
\begin{proof}
Starting from the definition of recoverability, we have
\begin{equation}\label{eq-lemma1proof}
\begin{aligned}
    % &\tr(O\rho) = \tr(\tilde{O}\channel(\rho))\quad \forall\rho\\
    % \Leftrightarrow\ 
    % &\tr(O\rho) = \tr(\channel^\dagger(\tilde{O})\rho)\quad \forall\rho\\
    % \Leftrightarrow\ 
    &O = \channel^\dagger(\tilde{O})\\
    \Leftrightarrow\ 
    &\channel^\dagger(\tilde{O})O=\identity_P\\
    \Leftrightarrow\  
    &\tr(\channel^\dagger(\tilde{O})O)=\dim P\\
    \Leftrightarrow\  
    &\tr(\Phi_\channel O^T\otimes \tilde{O})=1.
\end{aligned}
\end{equation}
 The second last equivalence follows from the fact that $\identity$ is the only Pauli operator with non-zero trace and the fact that $\channel^\dagger$ maps Pauli operators to Pauli operators. The last equivalence is due to 
 \begin{equation}
 \begin{aligned}
    \tr(\channel^\dagger(\tilde{O})O)=&\tr(\tilde{O}\channel (O))\\
    =&\eqfig{1cm}{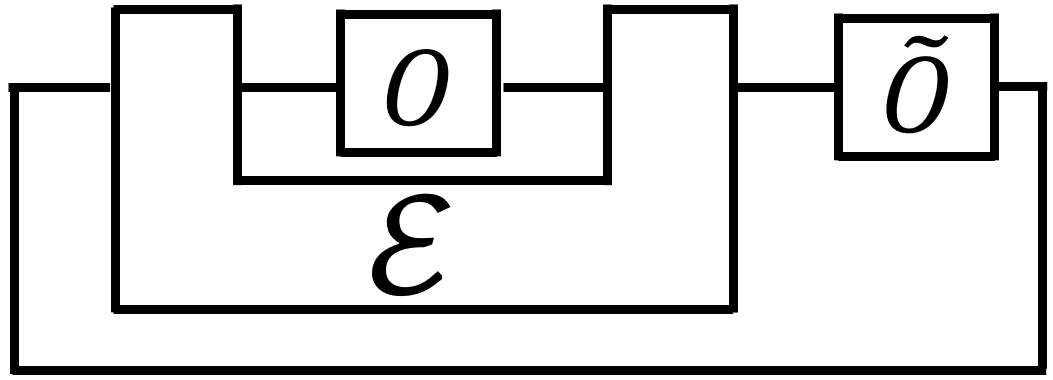}\\
    =&\eqfig{1.3cm}{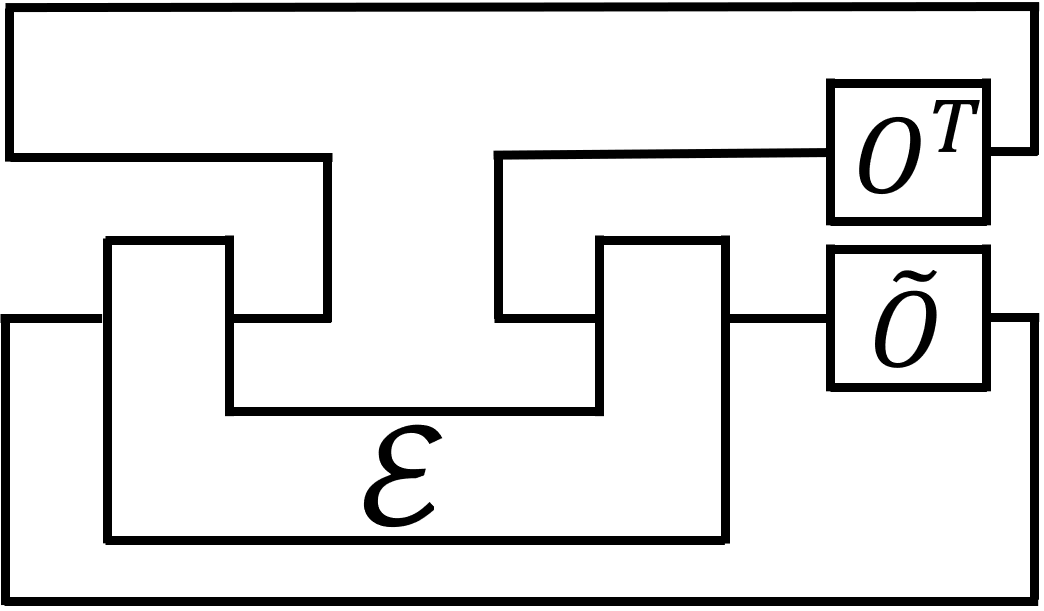}\\
    =&\dim P\cdot \tr(\Phi_\channel O^T\otimes \tilde{O}).
\end{aligned}    
 \end{equation}
 The last condition in \refeq{eq-lemma1proof} shows that $O^T\otimes \tilde{O}$ is a stabilizer.
 This completes the proof. 
\end{proof}

Next, we relate the amount of transferred operators to the mutual information, as summarized in the following lemma. 

\begin{lemma}
Given a stabilizer channel $\channel:\opspace(\hilbert_P)\rightarrow\opspace(\hilbert_Q)$ and two regions $A\subseteq P$, $S\subseteq Q$, we have $\log_2 H_{A\rightarrow S}(\channel) = I_{A:S}(\Phi_\channel) $.
\end{lemma}

\begin{proof}
Letting $S$ be the stabilizer group of $\Phi_\channel$. Define $G_S=G\cap \pauli_S$ and similarly $G_{AS}$. By lemma~\ref{lemma:Phistab}, we have:
\begin{equation}
    \log_2 H_{A\rightarrow B}(\channel)=\log_2|G_{AS}|-\log_2|G_{S}|.
\end{equation}
Here, $\log_2|G_{B}|$ is subtracted since the $\tilde{O}$ in lemma~\ref{lemma:Phistab} has the redundancy resulting from multiplying stabilizer operators in $G_S$. 

Let us recall the formula for entanglement entropy:
\begin{equation}
    S_{AS}=|A|+|S|-\log_2 |G_{AS}|,~~~S_{S}=|S|-\log_2 G_{S}.
\end{equation}
We obtain
\begin{equation}
    \log_2 H_{A\rightarrow S}(\channel)=|A|+|S|-S_{AS}-|S|+S_S= I_{A:S}(\Phi_\channel).
\end{equation}
The last equation is due to $S_A=|A|$ since $A$ is maximally mixed.  This completes the proof. 
\end{proof}

We are now ready to prove Thm.\ref{thm:channelinfo}. Recall that entanglement of the ensemble state $\Phi_{\channel}$ and that of each trajectory $\ket{\phi_{\textbf{m}}}$ are related by the following property of the conditional entropy:
\begin{equation}\label{eq:E5}
\begin{aligned}
    \E_{\textbf{m}}[S_X(\ket{\phi_{\textbf{m}}})]=S_{X|M}(\Phi_{\channel}) = S_{X\cup M}(\Phi_{\channel}) -S_M(\Phi_{\channel}).
\end{aligned}
\end{equation}
Here $X$ is any subregion of $P\cup R$, and $\E_{\textbf{m}}$ denotes averaging over all the trajectories:
\begin{equation}
\E_{\textbf{m}}[\cdots] \equiv \sum_{\textbf{m}}p_{\textbf{m}}(\cdots)_{\textbf{m}}.
\end{equation}
Moreover, for a Clifford monitored circuit, entanglement entropies of $\ket{\phi_{\textbf{m}}}$ do not depend on $\textbf{m}$ as long as $p_{\textbf{m}}\neq 0$.

We are now ready to derive \refeq{eq:monitored_relation}: 
\begin{equation}
\begin{aligned}
    I_{A:B}(\ket{\phi_{\textbf{m}}})
    =&(S_{A}
    + S_{B} - S_{A\cup B})(\ket{\phi_{\textbf{m}}})\\
    =& (S_{A_R|M}
    + S_{B_P|M}
    - S_{A_R\cup B_P|M})(\Phi_{\channel})\\
    =& (I_{A_R:(B_P\cup M)} - I_{A_R:M})(\Phi_{\channel})\\
    =& \log_2 H_{A\rightarrow B\cup M}(\channel) - \log_2 H_{A\rightarrow M}(\channel).
\end{aligned}
\end{equation}
Here the second equality follows from \refeq{eq:E5}, and the last equality follows from lemma \ref{lemma:Phistab}. This completes the proof of Thm.~\ref{thm:channelinfo}.

\section{Details about random polymer simulation}\label{ap:dpre-sim}
To perform numerical simulations of DPRE, we consider a lattice discretization of $D_V(\textbf{x}, \textbf{y})$ defined in \refeq{eq:randpolymer}. We replace the continuous spacetime manifold with a $L\times T$ grid, then change the continuous shortest path $\textbf{z}(\tau)$ into a discrete one $\textbf{z}_i$ on the grid: 
\begin{equation}
    D_V^{\textsf{discretized}}(\textbf{x}, \textbf{y}) = \min_{\substack{\textbf{z}: \textbf{z}_0=\textbf{x},\\ 
        \textbf{z}_{l(\textbf{z})}=
        \textbf{y}}}\left(\sum_{i=0}^{l(\textbf{z})-1} V_{(\textbf{z}_{i}, \textbf{z}_{i+1})}\right),
\end{equation}
where the integer valued $l(\textbf{z})$ is the total length of the path $\textbf{z}$. The discretized random potential $V_{(\textbf{z}, \textbf{z}')}$ is defined on the grid's links, and takes independent values at different locations. The problem, when formulated this way, is equivalent to looking for the ground state domain-wall of a random bond Ising model at zero temperature \cite{huse1985pinning}.

For simulations presented in \refsec{sec:domain-wall}, for each site $\textbf{z}$ we take $V_{\textbf{z}}$ to be a uniform distribution on the interval $(0,1)$. Since all $V$s are positive, the shortest path can be solved with the Dijkstra's algorithm using $O(LT\log(LT))$ time.

\end{document}